\documentclass[11pt, letterpaper, reqno]{amsart}
\usepackage[utf8]{inputenc} 	
\usepackage{microtype} 			
\usepackage{geometry} 			
\usepackage{amsmath}			
\usepackage{amsthm}		 		
\usepackage{amssymb}	 		
\usepackage{bm}					
\usepackage{mathrsfs}			
\usepackage{xcolor}				
\definecolor{darkcandyapp}{rgb}{0.6, 0.1, 0.1}
\definecolor{darkblue}{rgb}{0.2, 0.2, 0.6}
\definecolor{darkgreen}{rgb}{0.2, 0.4,.1}
\definecolor{mellowyellow}{rgb}{1,.8,.2}
\usepackage{tikz}				
\usepackage{tikz-3dplot}
\usepackage{pgfplots}
\pgfplotsset{compat=1.13}
\usetikzlibrary{decorations.markings,decorations.pathmorphing,shapes}
\usepackage[all]{xy}			
\usepackage{graphicx}			
\usepackage{subcaption}
\usepackage{floatrow}
\usepackage{enumitem} 			
\usepackage{booktabs}
\usepackage{array}
\usepackage{lmodern}			
\usepackage[T1]{fontenc}		
\usepackage[colorlinks=true, citecolor=cyan, urlcolor=magenta]{hyperref}
\usepackage{cleveref}
\usepackage[style=alphabetic, date=year, backref=true, firstinits=true, isbn=false]{biblatex}
\makeatletter
\renewbibmacro{in:}{}
\DeclareFieldFormat{pages}{#1}
\renewcommand*{\bibnamedash}{%
	\leavevmode\raise +0.6ex\hbox to 5.5ex{\hrulefill}.\space\space}

\InitializeBibliographyStyle{\global\undef\bbx@lasthash}

\newbibmacro*{bbx:savehash}{\savefield{fullhash}{\bbx@lasthash}}

\renewbibmacro*{author}{%
	\ifboolexpr{
		test \ifuseauthor
		and
		not test {\ifnameundef{author}}
	}
	{%
		\iffieldequals{fullhash}{\bbx@lasthash}
		{\bibnamedash\addcomma\space}
		{\printnames{author}}%
		\usebibmacro{bbx:savehash}%
		\iffieldundef{authortype}
		{}
		{%
			\setunit{\addcomma\space}%
			\usebibmacro{authorstrg}%
		}%
	}
	{\global\undef\bbx@lasthash}%
}
\makeatother
\newtheorem{propositionx}{Proposition}[section]
\newenvironment{proposition}
{\pushQED{\qed}\propositionx}
{\popQED\endpropositionx}

\newtheorem*{theoremx}{Theorem}
\newenvironment{theorem}
{\pushQED{\qed}\theoremx}
{\popQED\endtheoremx}

\newtheorem{lemmax}[propositionx]{Lemma}
\newenvironment{lemma}
{\pushQED{\qed}\lemmax}
{\popQED\endlemmax}
\newenvironment{lemmap}
{\pushQED{\qed}\lemmax}
{\popQED\endlemmax}

\theoremstyle{remark}
\newtheorem*{remarkx}{Remark}
\newenvironment{remark}
{\pushQED{\qed}\remarkx}
{\popQED\endremarkx}
\newtheorem*{examplex}{Example}
\newenvironment{example}
{\pushQED{\qed}\examplex}
{\popQED\endexamplex}

\newcommand{\dd}{\,\mathrm{d}}

\newcommand{\bbC}{\mathbb{C}}

\newcommand{\bbN}{\mathbb{N}}

\newcommand{\bbR}{\mathbb{R}}
\newcommand{\bbS}{\mathbb{S}}
\newcommand{\bbT}{\mathbb{T}}

\newcommand{\bbZ}{\mathbb{Z}}

\newcommand{\calF}{\mathcal{F}}
\newcommand{\calG}{\mathcal{G}}
\newcommand{\calH}{\mathcal{H}}
\newcommand{\calI}{\mathcal{I}}
\newcommand{\calJ}{\mathcal{J}}
\newcommand{\calK}{\mathcal{K}}

\newcommand{\calM}{\mathcal{M}}

\newcommand{\calS}{\mathcal{S}}

\newcommand{\calU}{\mathcal{U}}

\newcommand{\frakS}{\mathfrak{S}}

\newcommand{\bmalpha}{\bm{\alpha}}
\newcommand{\bmbeta}{\bm{\beta}}
\newcommand{\bmgamma}{\bm{\gamma}}

\title{The regularization of Dotsenko--Fateev integrals}
\author{Ethan Sussman}
\date{August 4th, 2023.}
\email{ethanws@mit.edu}
\address{Department of Mathematics, Massachusetts Institute of Technology, Massachusetts 02139-4307, USA}
\subjclass[2020]{Primary 32A20; Secondary 33C60, 33C90, 81T40}

\begin{document}

\begin{abstract}
	We discuss the regularization of certain hypergeometric integrals appearing in 2D CFT, a step needed in the construction of the BPZ minimal models via the Coulomb gas formalism. The method is a generalization of Pochhammer's regularization of the Euler Beta-function. The constructions of the relevant homology classes are inspired by a recent singular-geometric analysis of the Dotsenko--Fateev integrand. 
\end{abstract}

\maketitle

\tableofcontents

\section{Introduction}

For each $\ell,m,n\in \bbN$ not all zero, let 
\begin{equation} 
	\square^{\ell,m,n}_x = [-\infty,0]^\ell_{x_1,\cdots,x_\ell} \times [0,1]^m_{x_{\ell+1},\cdots,x_{\ell+m}} \times [1,\infty]^n_{x_{\ell+m+1},\cdots,x_N} \subseteq \overline{\bbR}^N, 
\end{equation} 
where $N=\ell+m+n$. 
In \cite{Sussman}, we showed that the \emph{Dotsenko--Fateev integral} \cite{DF2,Felder,FS1,FW}
\begin{equation}
	I_{\ell,m,n}(\bmalpha,\bmbeta,\bmgamma) = \int_{\square^{\ell,m,n}_x} \prod_{i=1}^N x_i^{\alpha_i} (1-x_i)^{\beta_i} \prod_{1\leq j< k \leq N} (x_k-x_j+i0)^{2\gamma_{j,k}} \dd x_1 \cdots \dd x_N,
	\label{eq:I} 
\end{equation}
defined initially for $\bmalpha=\{\alpha_i\}_{i=1}^N, \bmbeta=\{\beta_i\}_{i=1}^N\in \bbC^N$ and $\bmgamma = \{\gamma_{j,k}=\gamma_{k,j}\}_{1\leq j<k \leq N} \in \bbC^{N(N-1)/2}$ such that the integrand above (which is defined using the principal branch of the logarithm, with the branch cut being along the positive real axis) lies in $L^1(\square^{\ell,m,n}_x)$, admits an analytic continuation to almost all $(\bmalpha,\bmbeta,\bmgamma) \in \smash{\bbC^{N}\times \bbC^N\times \bbC^{N(N-1)/2}}$.
More precisely, there exists a locally finite collection $\calH_{\ell,m,n}$ of complex affine hyperplanes' worth of parameters such that $I_{\ell,m,n}$ can be extended to an analytic function 
\begin{equation}
	I_{\ell,m,n} : (\bbC^{N}\times \bbC^N\times \bbC^{N(N-1)/2})\backslash (\cup_{H\in \calH_{\ell,m,n}} H)  \to \bbC. 
\end{equation}
Since the domain is connected, this extension is unique. Moreover, $\calH_{\ell,m,n}$ does not contain any affine hyperplanes of the form $\{\gamma=c\}$ for $c\in \bbC$. This final observation is critical to the application of Dotsenko--Fateev integrals in the CFT literature. This application, part of the celebrated Coulomb-gas formalism, has yet to be made fully rigorous, in part because of difficulties in analyzing the the singularity structure of the $I_{\ell,m,n}$.

We re-prove the stated result here using an argument inspired by the treatment of Selberg-like integrals in \cite{KT2, KT1}. In fact, the Dotsenko--Fateev integrals are sums of Selberg-like integrals, so the meromorphic continuation of the $I_{\ell,m,n}$ to \emph{almost all} $\bmalpha,\bmbeta,\bmgamma$ is a corollary. 
However, this does not yield the fact that $\calH_{\ell,m,n}$ does not contain any affine hyperplanes of the form $\{\gamma=c\}$. Indeed, the resultant meromorphic continuation has an apparent singularity when any of the $\gamma_{j,k}$ are taken equal to $-1$, which is precisely the value appearing in the CFT literature when working with both possible types of screening charges. 
For this application, generic values of the parameters do not suffice. 
It is a fact of life that the apparent singularity is removable, but one must prove it. The argument below is one means of doing so. 

One reason to prefer the regularization method here over that in \cite{Sussman} is that it can be applied verbatim to the formal integrals in \cite{Felder,FS1,FS2} defining the screened vertex \emph{operators}, out of which the chiral algebras of the Belavin-Polyakov-Zamolodchikov (BPZ) minimal models \cite{BPZ} are supposed  to be constructed. Our earlier paper applied only to the 3-point coefficients.
Our new regularization yields a definition of the screened vertex operators as honest bounded linear maps between appropriate Hilbert spaces. 
On the analytic side, this relies on the treatment of the unscreened case in \cite{Boenkost,Bovier}.

In the $N=1$ case, the method applied here, as well as in \cite{KT2, KT1}, becomes Pochhammer's regularization of of the Euler Beta function 
\begin{equation}
	B(\alpha+1,\beta+1) = \int_0^1 x^\alpha(1-x)^\beta \dd x = \frac{\Gamma(1+\alpha)\Gamma(1+\beta)}{\Gamma(2+\alpha+\beta)}. 
\end{equation}
Indeed, $B$ is $I_{0,1,0}$ up to a conventional $+1$ shift of its arguments. The \emph{Pochhammer contour} is the element 
\begin{equation} 
	b^{-1} a^{-1} ba \in \pi_1(\bbC\backslash\{0,1\}),
\end{equation} 
where $a,b$ are the generators of $\pi_1(\bbC\backslash \{0,1\})$ corresponding to one counterclockwise circuit around $0,1$ respectively. This lifts to a closed contour $\Gamma$ in a cover of $\bbC\backslash \{0,1\}$ on which $z^\alpha(1-z)^\beta$ defines a single-valued function. Then, choosing the branch of this lift appropriately, 
\begin{equation}
	\int_{\Gamma} z^\alpha(1-z)^\beta \dd z = -4\sin(\pi \alpha)\sin(\pi \beta) B(\alpha+1,\beta+1), 
\end{equation} 
with different choices of branch giving the same result up to a phase (which does not matter as far as the rest of the argument is concerned). Rearranging, 
\begin{equation}
	B(\alpha+1,\beta+1) = -\frac{1}{4\sin(\pi \alpha)\sin(\pi \beta)} \int_{\Gamma} z^\alpha(1-z)^\beta \dd z
	\label{eq:misc_007}
\end{equation}
as long as we are not dividing by $0$.
A priori (that is, without already knowing the analyticity properties of $B$), the integral $\int_\Gamma z^\alpha(1-z)^\beta \dd z$ manifestly makes sense \emph{for any} $\alpha,\beta \in \bbC$ and defines an entire function of these parameters. The meromorphic singularities of $B$ therefore all appear in the prefactor $1/\sin(\pi\alpha)\sin(\pi\beta)$ on the right-hand side of \cref{eq:misc_007}.

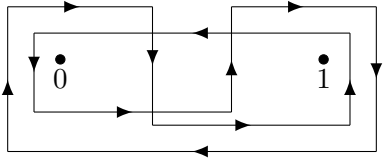
\begin{figure}[h]
	\begin{tikzpicture}[scale=3.5,  decoration={
			markings,
			mark=at position 0.5 with {\arrow[scale=1.5,>=latex]{>}}}]
		\coordinate (0) at (1.2,.2);
		\coordinate (1) at (1.2,-.35);
		\coordinate (2) at (-.2,-.35);
		\coordinate (3) at (-.2,+.2);
		\coordinate (4) at (.35,+.2);
		\coordinate (5) at (.35,-.25);
		\coordinate (6) at (1.1,-.25);
		\coordinate (7) at (1.1,.1);
		\coordinate (8) at (-.1,.1);
		\coordinate (9) at (-.1,-.2);
		\coordinate (10) at (.65,-.2);
		\coordinate (11) at (.65,.2);
		\draw[postaction={decorate}] (0) -- (1);
		\draw[postaction={decorate}] (1) -- (2);
		\draw[postaction={decorate}] (2) -- (3);
		\draw[postaction={decorate}] (3) -- (4);
		\draw[postaction={decorate}] (4) -- (5);
		\draw[postaction={decorate}] (5) -- (6);
		\draw[postaction={decorate}] (6) -- (7);
		\draw[postaction={decorate}] (7) -- (8);
		\draw[postaction={decorate}] (8) -- (9);
		\draw[postaction={decorate}] (9) -- (10);
		\draw[postaction={decorate}] (10) -- (11);
		\draw[postaction={decorate}] (11) -- (0);
		\filldraw[color=black] (0,0) circle (.5pt) node[below] {$0$};
		\filldraw[color=black] (1,0) circle (.5pt) node[below] {$1$};
	\end{tikzpicture}
	\caption{A contour in $\bbC\backslash \{0,1\}$ homotopic to the Pochhammer contour.}
\end{figure}

Since the Beta function is easily computed in terms of the Gamma function, there are more direct ways of producing its meromorphic extension. However, for the general Dotsenko--Fateev integrals considered here, we do not have (even conjecturally) a formula. In \cite{DF3}, Dotsenko and Fateev do claim a formula for the integrals showing up as the 3-point coefficients between primary fields in the BPZ minimal models, but making the computation rigorous seems to require already knowing some variant of the result proven here. Moreover, when computing matrix elements involving in- and out-states besides the vacuum, more general integrals than those considered by Dotsenko--Fateev -- but still of the form \cref{eq:I} -- appear. For these, Dotsenko--Fateev do not even conjecture a formula, and their argument cannot be modified to suggest one. So, it is desirable to have a means of producing meromorphic extensions that does not rely on explicit formula holding on some small subspace worth of parameters.

The integrand 
\begin{equation} 
	\omega(\bmalpha,\bmbeta,\bmgamma) =  \prod_{i=1}^N z_i^{\alpha_i} (1-z_i)^{\beta_i} \prod_{1\leq j< k \leq N} (z_k-z_j)^{2\gamma_{j,k}} \dd z_1 \wedge \cdots \wedge \dd z_N 
\end{equation} 
of the Dotsenko--Fateev integral, the ``$+i0$'' having been removed, is a multi-valued analytic $N$-form on the moduli space $\calM_N(0,\infty)$, where, for $0\leq r<1<R \leq \infty$,
\begin{equation}
	\calM_N(r,R) = 
	\{(z_1,\cdots,z_N)\in (\bbC\backslash \{0,1\})^N : z_1,\cdots,z_N\text{ distinct, s.t.\ } r< |z_j| < R\text{ for all }j\}.
\end{equation}
This is the moduli space of $N$ pairwise distinct elements of the punctured annulus $\{z\in \bbC : r< |z|<R , z\neq 1\}$.
Alternatively, we can consider $\omega(\bmalpha,\bmbeta,\bmgamma)$ as a \emph{single-valued} analytic function on the monodromy cover
\begin{equation}
	\widehat{\calM}_N(r,R) = \widetilde{\calM}_N(r,R) \backslash [\pi_1(\calM_N(r,R) ),\pi_1(\calM_N(r,R) ) ],
\end{equation}
which is also a complex manifold. Thus, we can consider $\omega(\bmalpha,\bmbeta,\bmgamma)$ as an element of $\smash{\Omega^{N}(\widehat{\calM}_N(0,\infty))}$, depending analytically on $\bmalpha,\bmbeta,\bmgamma$.
As a consequence of analyticity in $z_1,\cdots,z_N$, this form is closed and therefore defines an element of de Rham cohomology. This cohomology class is nonzero.

The problem of regularizing the Dotsenko--Fateev integral $I_{\ell,m,n}$ is therefore to find a \emph{multi-contour} 
\begin{equation} 
	\Gamma_{\ell,m,n} \in H_{N}(\widehat{\calM}_N(0,\infty);\bbZ)
\end{equation} 
such that
\begin{equation} 
	\int_{\Gamma_{\ell,m,n}} \omega(\bmalpha,\bmbeta,\bmgamma) \propto I_{\ell,m,n}(\bmalpha,\bmbeta,\bmgamma),
	\label{eq:misc_012} 
\end{equation} 
where the ``$\propto$'' denotes proportionality up to constants and trigonometric functions of $\bmalpha,\bmbeta,\bmgamma$. The contribution of this work is the construction of such an element. This has been elusive since \cite{FS2}, in which a similar result was a necessary element of their main argument. Lacking a proof, it was conjectured therein.

The class $\Gamma_{\ell,m,n}$ has the form 
\begin{equation}
	\Gamma_{\ell,m,n} = (\iota_{\ell,m,n})_{*}([\mathsf{2}(A_{\ell,m,n})] ), 
	\label{eq:misc_013}
\end{equation}
where $\mathsf{2}(A_{\ell,m,n})$ is a compact, orientable, smoothable $C^0$-manifold without boundary constructed in a relatively simple way out of explicit smooth manifolds-with-corners $A_{\ell,m,n}$ defined in \cite[\S2.2]{Sussman}, $[\mathsf{2}(A_{\ell,m,n})] \in H_N(\mathsf{2}(A_{\ell,m,n});\bbZ)$ is the fundamental class of this manifold, and
\begin{equation} 
\iota_{\ell,m,n}: \mathsf{2}(A_{\ell,m,n}) \to \widehat{\calM}_N(0,\infty)
\label{eq:misc_14}
\end{equation} 
is some continuous map defined below. Most of the work below consists in defining $\iota_{\ell,m,n}$.

When $N=1$, the manifold $\mathsf{2}(A_{\ell,m,n}) \in \{ \mathsf{2}(A_{1,0,0}), \mathsf{2}(A_{0,1,0}), \mathsf{2}(A_{0,0,1})\}$ is just $\bbS^1$, and $\iota_{\ell,m,n}$ is just the lift of the Pochhammer contour, up to homotopy (and possibly a choice of branch).

\begin{figure}[t!]
	\begin{center} 
		\tdplotsetmaincoords{75}{110}
		\begin{tikzpicture}[scale=2.25,tdplot_main_coords]
			\draw[opacity=0] (1.5,0,0) -- (0,1.5,0) -- (0,0,1.5); 
			\draw[fill=gray,fill opacity = .1] (1,.2,0) -- (1,.8,0) -- (1,1,.2) -- (1,1,1) -- (1,.2,1) -- cycle;%
			\draw[fill=gray,fill opacity = .1] (1,1,.2) -- (1,1,1) -- (.2,1,1) -- (0,1,.8) -- (0,1,.2) -- cycle;%
			\draw[fill=gray,fill opacity = .1] (1,1,1) -- (.2,1,1) -- (.2,0,1) -- (.8,0,1) -- (1,.2,1) -- cycle;%
			\draw[fill=gray,fill opacity = .1] (1,1,.2) -- (0,1,.2) -- (0,.8,0) -- (1,.8,0) -- cycle;%
			\draw[fill=gray,fill opacity = .1] (1,.2,0) -- (.8,0,0) -- (.8,0,1) -- (1,.2,1) -- cycle;%
			\draw[fill=gray,fill opacity = .1, draw=none] (.2,1,1) -- (0,1,.8) -- (0,0,.8) -- (.2,0,1) -- cycle;%
			\draw[fill=gray,fill opacity = .1, dashed] (0,0,.8) -- (.2,0,1) -- (.8,0,1) -- (.8,0,0) -- (0,0,0) -- cycle;%
			\draw[fill=gray,fill opacity = .1, dashed] (0,0,0) -- (0,0,.8) -- (0,1,.8) -- (0,1,.2) -- (0,.8,0) -- cycle;%
			\draw[fill=gray,fill opacity = .1, draw=none] (0,0,0) -- (.8,0,0) -- (1,.2,0) -- (1,.8,0) -- (0,.8,0) -- cycle;%
			\node (lb) at (0,.2,1.2) {$A_{1,1,1}$};%
		\end{tikzpicture}
		\begin{tikzpicture}[scale=2.25,tdplot_main_coords]
			\draw[opacity=0] (1.5,0,0) -- (0,1.5,0) -- (0,0,1.5); 
			\draw[fill=gray,fill opacity = .1] (1,.2,.1) -- (1,.1,.2) -- (1,.1,1) -- (1,.8,1) -- (1,1,.8) -- (1,1,.1) -- cycle;%
			\draw[fill=gray,fill opacity = .1] (1,1,.1) -- (.9,1,0) -- (0,1,0) -- (0,1,.8) -- (1,1,.8) -- cycle;%
			\draw[fill=gray,fill opacity = .1] (1,.8,1) -- (0,.8,1) -- (0,0,1) -- (.9,0,1) -- (1,.1,1) -- cycle;%
			\draw[fill=gray,fill opacity = .1] (1,.8,1) -- (0,.8,1) -- (0,1,.8) -- (1,1,.8) -- cycle;%
			\draw[fill=gray,fill opacity = .1] (1,1,.1) -- (.9,1,0) -- (.9,.2,0) -- (1,.2,.1) -- cycle;%
			\draw[fill=gray,fill opacity = .1] (1,.1,1) -- (.9,0,1) -- (.9,0,.2) -- (1,.1,.2) -- cycle;%
			\draw[fill=gray,fill opacity = .1] (1,.2,.1) -- (1,.1,.2) -- (.9,0,.2) -- (.8,0,.1) -- (.8,.1,0) -- (.9,.2,0) -- cycle;%
			\draw[fill=gray,fill opacity = .1, draw=none] (0,.1,0) -- (0,0,.1) -- (.8,0,.1) -- (.8,.1,0) -- cycle;%
			\draw[fill=gray,fill opacity = .1, dashed] (.9,.2,0) -- (.9,1,0) -- (0,1,0) -- (0,.1,0) -- (.8,.1,0) -- cycle; %
			\draw[fill=gray,fill opacity = .1, dashed] (.8,0,.1) -- (.9,0,.2) -- (.9,0,1) -- (0,0,1) -- (0,0,.1) -- cycle; %
			\draw[fill=gray,fill opacity = .1, dashed] (0,0,.1) -- (0,.1,0) -- (0,1,0) -- (0,1,.8) -- (0,.8,1) -- (0,0,1) -- cycle; %
			\node (lb) at (0,.2,1.2) {$A_{1,2,0}$}; %
		\end{tikzpicture}
		\begin{tikzpicture}[scale=2.25,tdplot_main_coords]
			\draw[opacity=0] (1.5,0,0) -- (0,1.5,0) -- (0,0,1.5); 
			\draw[fill=gray,fill opacity = .1] (1,.9,.8) -- (1,.8,.9) -- (.9,.8,1) -- (.8,.9,1) -- (.8,1,.9) -- (.9,1,.8) -- cycle;%
			\draw[fill=gray,fill opacity = .1] (1,.9,.8) -- (.9,1,.8) -- (.9,1,0) -- (1,.9,0) -- cycle;%
			\draw[fill=gray,fill opacity = .1] (1,.8,.9) -- (.9,.8,1) -- (.9,0,1) -- (1,0,.9) -- cycle;%
			\draw[fill=gray,fill opacity = .1] (.8,.9,1) -- (.8,1,.9) -- (0,1,.9) -- (0,.9,1) -- cycle;%
			\draw[fill=gray,fill opacity = .1] (.8,1,.9) -- (.9,1,.8) -- (.9,1,0) -- (.1,1,0) -- (0,1,.1) -- (0,1,.9) -- cycle;%
			\draw[fill=gray,fill opacity = .1] (1,.9,.8) -- (1,.8,.9) -- (1,0,.9) -- (1,0,.1) -- (1,.1,0) -- (1,.9,0) -- cycle;%
			\draw[fill=gray,fill opacity = .1] (.9,.8,1) -- (.8,.9,1) -- (0,.9,1) -- (0,.1,1) -- (.1,0,1) -- (.9,0,1) -- cycle;%
			\draw[fill=gray,fill opacity = .1,dashed] (0,.1,.2) -- (0,.2,.1) -- (.1,.2,0) -- (.2,.1,0) -- (.2,0,.1) -- (.1,0,.2) -- cycle;%
			\draw[fill=gray,fill opacity = .1,draw=none] (0,.1,.2) -- (.1,0,.2) -- (.1,0,1) -- (0,.1,1) -- cycle;%
			\draw[fill=gray,fill opacity = .1,draw=none] (0,.2,.1) -- (.1,.2,0) -- (.1,1,0) -- (0,1,.1) -- cycle;%
			\draw[fill=gray,fill opacity = .1,draw=none] (.2,.1,0) -- (.2,0,.1) -- (1,0,.1) -- (1,.1,0) -- cycle;%
			\draw[fill=gray,fill opacity = .1,dashed] (.2,0,.1) -- (.1,0,.2) -- (.1,0,1) -- (.9,0,1) -- (1,0,.9) -- (1,0,.1) -- cycle;%
			\draw[fill=gray,fill opacity = .1,dashed] (0,.1,.2) -- (0,.2,.1) -- (0,1,.1) -- (0,1,.9) -- (0,.9,1) -- (0,.1,1) -- cycle;%
			\draw[fill=gray,fill opacity = .1,dashed] (.1,.2,0) -- (.2,.1,0) -- (1,.1,0) -- (1,.9,0) -- (.9,1,0) -- (.1,1,0) -- cycle;%
			\node (lb) at (0,.2,1.2) {$A_{0,3,0}$};
		\end{tikzpicture}
	\end{center} 
	\caption{The three mwcs $A_{1,1,1}$, $A_{1,2,0}$, $A_{0,3,0}$. }
	\label{fig:A}
\end{figure}
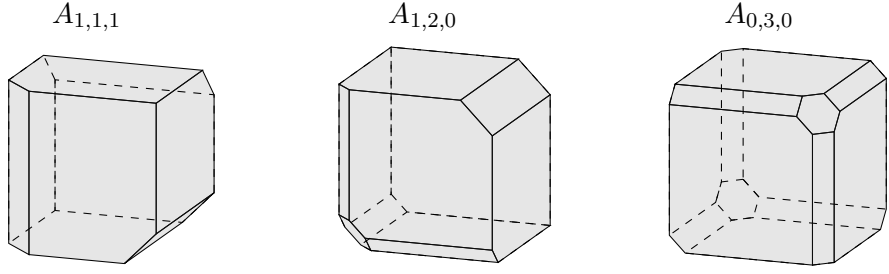

As the dimension $N=\ell+m+n$ increases, the topological complexity of these objects increases rapidly. For instance, as computed below, the manifolds $\mathsf{2}(A_{0,2,0})\cong \mathsf{2}(A_{2,0,0}) \cong \mathsf{2}(A_{0,0,2})$ are already 17-holed tori. The somewhat simpler manifolds $\mathsf{2}(A_{1,1,0})\cong \mathsf{2}(A_{1,0,1}) \cong \mathsf{2}(A_{0,1,1})$ are 5-holed tori. For $N\geq 3$, it does not seem possible to give a more succinct description of these objects than the constructions themselves. 

We now state the precise version of \cref{eq:misc_012} that is our main theorem. For each $\calS\subseteq \{1,\cdots,N\}$, let  
\begin{equation}
	\alpha_{\calS} = \sum_{j\in \calS} \alpha_j + 2 \sum_{1\leq j<k\leq N,\, j,k\in \calS} \gamma_{j,k}, \qquad 
	\beta_{\calS} = \sum_{j\in \calS} \beta_j + 2 \sum_{1\leq j<k\leq N,\, j,k\in \calS} \gamma_{j,k},
	\label{eq:abdef} 
\end{equation}
\begin{equation}
	\zeta_{\calS} = - \sum_{j\in \calS}(\alpha_j+\beta_j) - 2 \sum_{\substack{1\leq j<k \leq N \\ j\in \calS\text{ or }k\in \calS}} \gamma_{j,k}.
	\label{eq:zetadef}
\end{equation}
Then, our main result is: 
\begin{theorem}
	For each $r,R\in [0,\infty]$ such that $0\leq r<1<R\leq \infty$, there exists a continuous map $\iota_{\ell,m,n}[r,R] :\mathsf{2}(A_{\ell,m,n}) \to  \widehat{\calM}_N(r,R)$ such that 
	\begin{multline}
		\int_{\iota_{\ell,m,n}[r,R]_*([\mathsf{2}(A_{\ell,m,n})]) }  \omega(\bmalpha,\bmbeta,\bmgamma) = \Big[ \prod_{\varnothing\subsetneq \calS\subseteq \{1,\cdots,\ell+m\}} 2i\sin(\pi \alpha_{\calS} )  \Big] \Big[ \prod_{\varnothing \subsetneq \calS \subseteq \{\ell+1,\cdots,N\}} 2i\sin(\pi \beta_{\calS})  \Big] \\ \times  \Big[ \prod_{\varnothing\subsetneq \calS \subseteq \{1,\cdots,\ell\}\cup \{\ell+m+1,\cdots,N\} } 2i\sin(\pi \zeta_{\calS})  \Big] I_{\ell,m,n}(\bmalpha,\bmbeta,\bmgamma)
		\label{eq:fundamental}
	\end{multline}
	for all $\bmalpha,\bmbeta\in \bbC^N$ and $\bmgamma \in \bbC^{N(N-1)/2}$ for which $I_{\ell,m,n}(\bmalpha,\bmbeta,\bmgamma)$ is defined. 
	\label{thm}
\end{theorem}

\begin{remark}
	In fact, we will construct $\iota_{\ell,m,n}[r,R]$ such that it is smooth with respect to some explicit smooth structure on $\mathsf{2}(A_{\ell,m,n})$. The construction should output an immersion, at least for generic values of the parameters appearing, but checking this is beyond the scope of this paper. 
	Hence, $\Gamma_{\ell,m,n}$ is the homology class of an immersed $C^\infty$-submanifold. 
\end{remark}
 
The meromorphic extension of the Dotsenko--Fateev integrals follows.

\section{Doubles of manifolds-with-corners}

Suppose that $M$ is a $C^\infty$-manifold-with-corners (mwc) in the sense of Melrose \cite{MelroseMWC,MelroseCorners}. Roughly, this means that $M$ is a smooth space modeled on $[0,\infty)^K \times \bbR^J$ for some $J,K\in \bbN$, in which boundary components are all embedded submanifolds. 
Just as out of any manifold-with-boundary one can construct its double as a $C^0$-manifold, from $M$ can be formed a canonical $C^0$-manifold without boundary 
\begin{equation}
	\mathsf{2}(M) = \Big(\bigcup_{\calF\subseteq  \mathcal{F}(M) } \{ (\calF, p) :p\in M\}\Big)/\!\sim\; = \Big( \bigsqcup_{\calF\subseteq \calF(M)} M \Big) /\!\sim 
\end{equation}
by gluing together $2^{|\calF(M)|}$ copies of $M$ along its various facets. Here, $\calF(M)$ is the set of facets (i.e.\ boundary hypersurfaces; the term ``face'' can be used to refer to boundary components of arbitrary codimension $\geq 1$) of $M$, which we will always assume are connected, and the equivalence relation $\sim$ is defined by 
\begin{equation}
	(\calF,p) \sim (\calF',p') \iff 
	\begin{cases}
		p=p' \text{ and } \\ 
		p\in \bigcap_{\mathrm{f}\in (\calF \Delta \calF')} \mathrm{f}, 
	\end{cases}
\end{equation}
where $\calF \Delta \calF' \subseteq \calF(M)$ is the symmetric difference between the sets $\calF,\calF'$. 
The topology on $\mathsf{2}(M)$ is the quotient topology, and it can be checked that it has the structure of a $C^0$-manifold. The map $\Pi:\mathsf{2}(M)\to M$ given by $[(\calF,p)]_\sim \mapsto p$ is continuous. 
Melrose calls $\mathsf{2}(M)$ the \emph{double} of $M$, even though $\mathsf{2}(M)$ consists of more than two sheets of $M$ if $M$ has more than one boundary hypersurface. For each $\calF\subseteq \calF(M)$, we can regard $\{[(\calF,p)]_\sim:p\in M\}$ as a smooth mwc diffeomorphic to $M$ via $\Pi$. In this sense, $\mathsf{2}(M)$ is piecewise-smooth.

If $M$ is compact, then $\mathsf{2}(M)$ is a compact topological manifold without boundary. 
If $M$ is already a manifold without boundary, then $\mathsf{2}(M)$ is just $M$ itself, with the smooth structure forgotten.

If $M$ is a manifold-with-boundary with a connected boundary, then $\mathsf{2}(M)$ is just the usual doubling of $M$ across the lone boundary, except that we are not yet endowing it with a smooth structure (as the smooth structure is actually not canonical).

\begin{example}
	Consider the interval $I=[0,1]$, which we regard as a manifold-with-corners with two different boundary components, the singletons $\{0\}$ and $\{1\}$. Then, $\mathsf{2}([0,1])\cong\bbS^1$. Note that $\smash{\#\Pi^{-1}(p)}=4$ for each $p\in (0,1)$, which is different from what happens when one doubles $I$ considered as a manifold-with-boundary with a single disconnected boundary hypersurface.
\end{example}
Since $\mathsf{2}(X\times Y)\cong\mathsf{2}(X)\times \mathsf{2}(Y)$ for any mwcs $X,Y$, the previous example implies that the double $\mathsf{2}([0,1]^N)$ of the $N$-cube is given by $\mathsf{2}([0,1]^N)\cong \bbT^N$, for any $N\in \bbN^+$. 

\begin{example}
	If $M \subset \bbR^2$ is $\{(x,y)\in \bbR^2: |y| \leq 1-x^2\}$, which we consider as a mwc with two boundary edges 
	\begin{equation} 
		\mathrm{e}_\pm=\{(x,y)\in \bbR^2: |x|\leq1\text{ and }\pm y = 1-x^2\},
	\end{equation} 
	then $\mathsf{2}(M)\cong \bbS^2$, and this is assembled from four copies of the original mwc in a manner reminiscent of the paneling of a standard American football. 
\end{example}

\begin{example}
	If $M$ is a convex $N$-gon for $N\geq 3$ (which is a sub-mwc of $\mathbb{R}^2$), then $\mathsf{2}(M)$ is the orientable surface of genus $g=\smash{N2^{N-3}-2^{N-1}+1}$. Indeed, it is orientable (see below) and constructed as a cell complex of $2^N$ cells, each of which is a copy of $M$. This cell decomposition has $F=2^N$ facets, $E=N 2^{N-1}$ edges, and $V=N 2^{N-2}$ vertices, so 
	\begin{equation} 
		g=1-\frac{1}{2}\chi(\mathsf{2}(M)) = 1 -  \frac{F-E+V}{2} =  N2^{N-3}-2^{N-1}+1. 
	\end{equation}  
	For example, the double $\mathsf{2}(\triangle)$ of a triangle $\triangle$ is a topological sphere (realized as an octahedron), the double of a square is a torus, as already seen, the double of the pentagon $A_{1,1,0}$ is a 5-holed torus, and the double of the hexagon $A_{0,2,0}$ is a 17-holed torus.
	
	The genus is increasing exponentially in $N$.
\end{example}

\begin{figure}
	\begin{center}
		\begin{tikzpicture}[scale=2.7, rotate around y=15, rotate around z=0, rotate around x=-135]
			\coordinate (t) at (0,0,.707);
			\coordinate (b) at (0,0,-.707);
			\coordinate (1) at (.5,.5,0);
			\coordinate (2) at (-.5,.5,0);
			\coordinate (3) at (-.5,-.5,0);
			\coordinate (4) at (.5,-.5,0);
			\draw[darkblue, dashed] (t) -- (2);
			\draw[mellowyellow, dashed] (1) -- (2);
			\draw[darkcandyapp, dashed] (t) -- (1);
			\fill[fill=gray!10, opacity=.5] (t) -- (1) -- (2) -- cycle;
			\fill[fill=gray!10, opacity=.5] (t) -- (2) -- (3) -- cycle;
			\draw[darkgray, dashed]  (-.33,0,.233) -- (-.5,0,.4) node[above left] {$\{{\color{mellowyellow}\bullet }, {\color{darkcandyapp}\bullet }\}$};
			\draw[darkgray, dashed]  (.33,0,.233) -- (.5,0,.4) node[above right] {$\{{\color{mellowyellow}\bullet }, {\color{darkblue}\bullet }\}$};
			\fill[fill=gray!10, opacity=.5] (t) -- (4) -- (1) -- cycle;
			\fill[fill=gray!10, opacity=.5] (b) -- (1) -- (2) -- cycle;
			\fill[fill=gray!10, opacity=.5] (b) -- (2) -- (3) -- cycle;
			\fill[fill=gray!10, opacity=.5] (b) -- (3) -- (4) -- cycle;
			\fill[fill=gray!10, opacity=.5] (b) -- (4) -- (1) -- cycle;
			\fill[fill=gray!10, opacity=.5] (t) -- (3) -- (4) -- cycle;
			\draw[darkcandyapp] (b) -- (1);
			\draw[darkcandyapp] (3) -- (b);
			\draw[mellowyellow] (1) -- (4);
			\draw[mellowyellow] (2) -- (3);
			\draw[mellowyellow] (3) -- (4);
			\draw[darkcandyapp] (3) -- (t);
			\draw[darkblue] (2) -- (b);
			\draw[darkblue] (4) -- (b);
			\draw[darkblue] (4) -- (t);
			\node (z) at (0,-.33,.233) {$\{{\color{mellowyellow}\bullet }\}$};
			\node (y) at (0,-.33,-.233) {$\varnothing$};
			\draw[darkgray, dashed]  (.33,0,-.233) -- (.5,0,-.4) node[below right] {$\{{\color{darkblue}\bullet }\}$};
			\draw[darkgray, dashed]  (-.33,0,-.233) -- (-.5,0,-.4) node[below left] {$\{{\color{darkcandyapp}\bullet }\}$};
		\end{tikzpicture}
		\quad
		\begin{tikzpicture}[scale=1.5, rotate around y=1, rotate around z=0, rotate around x=2]
			\coordinate (1o) at (-1,-1,0);
			\coordinate (2o) at (+1,-1,0);
			\coordinate (3o) at (+1,+1,0);
			\coordinate (4o) at (-1,+1,0);
			\coordinate (1i) at (-.5,-.5,0);
			\coordinate (2i) at (+.5,-.5,0);
			\coordinate (3i) at (+.5,+.5,0);
			\coordinate (4i) at (-.5,+.5,0);
			\coordinate (14i) at (-.5,.08,0);
			\coordinate (12i) at (.15,-.5,0);
			\coordinate (1ot) at (-1,-1,1);
			\coordinate (2ot) at (+1,-1,1);
			\coordinate (3ot) at (+1,+1,1);
			\coordinate (4ot) at (-1,+1,1);
			\coordinate (1it) at (-.5,-.5,1);
			\coordinate (2it) at (+.5,-.5,1);
			\coordinate (3it) at (+.5,+.5,1);
			\coordinate (4it) at (-.5,+.5,1);
			\draw[dashed]  (-1,0,.5) -- (-1.2,0,.5) node[left] {$\{{\color{mellowyellow}\bullet },{\color{darkblue}\bullet },{\color{darkgreen}\bullet } \}$};
			\draw[dashed] (.75,0,0) -- (.75,0,-.5) -- (1.5,0,-.5) node[right] {$\{{\color{mellowyellow}\bullet} ,{\color{darkcandyapp}\bullet } \}$};
			\fill[fill=gray!10, opacity=.5] (1i) -- (2i) -- (2o) -- (1o) -- cycle;
			\fill[fill=gray!10, opacity=.5] (1i) -- (4i) -- (4o) -- (1o) -- cycle;
			\fill[fill=gray!10, opacity=.5] (3i) -- (2i) -- (2o) -- (3o) -- cycle;
			\fill[fill=gray!10, opacity=.5] (3i) -- (4i) -- (4o) -- (3o) -- cycle;
			\draw[darkcandyapp, dashed] (4it) -- (4i) -- (4o);
			\draw[darkgreen, dashed] (1i) -- (1it) -- (1ot);
			\draw[darkcandyapp, dashed] (2o) -- (2i) -- (2it);
			\draw[darkgreen, dashed] (3o) -- (3i) -- (3it);
			\draw[mellowyellow, dashed] (14i) -- (4i) -- (3i) -- (2i) -- (12i);
			\draw[darkblue, dashed] (4o) -- (1o) -- (2o);
			\fill[fill=gray!10, opacity=.5] (1ot) -- (1o) -- (2o) -- (2ot) -- cycle;
			\fill[fill=gray!10, opacity=.5] (3ot) -- (3o) -- (2o) -- (2ot) -- cycle;
			\fill[fill=gray!10, opacity=.5] (1ot) -- (1o) -- (4o) -- (4ot) -- cycle;
			\fill[fill=gray!10, opacity=.5] (3ot) -- (3o) -- (4o) -- (4ot) -- cycle;
			\fill[fill=gray!10, opacity=.5] (1it) -- (1i) -- (2i) -- (2it) -- cycle;
			\fill[fill=gray!10, opacity=.5] (3it) -- (3i) -- (2i) -- (2it) -- cycle;
			\fill[fill=gray!10, opacity=.5] (1it) -- (1i) -- (4i) -- (4it) -- cycle;
			\fill[fill=gray!10, opacity=.5] (3it) -- (3i) -- (4i) -- (4it) -- cycle;
			\fill[fill=gray!10, opacity=.5] (1it) -- (2it) -- (2ot) -- (1ot) -- cycle;
			\fill[fill=gray!10, opacity=.5] (1it) -- (4it) -- (4ot) -- (1ot) -- cycle;
			\fill[fill=gray!10, opacity=.5] (3it) -- (2it) -- (2ot) -- (3ot) -- cycle;
			\fill[fill=gray!10, opacity=.5] (3it) -- (4it) -- (4ot) -- (3ot) -- cycle;
			\draw[darkblue] (2o) -- (3o) -- (4o);
			\draw[darkgreen] (3it) -- (3ot) -- (3o);
			\draw[mellowyellow] (14i) -- (1i) -- (12i);
			\draw[mellowyellow] (1ot) -- (2ot) -- (3ot) -- (4ot) -- cycle;
			\draw[darkblue] (1it) -- (2it) -- (3it) -- (4it) -- cycle;
			\draw[darkgreen] (1ot) -- (1o) -- (1i);
			\draw[darkcandyapp] (4it) -- (4ot) -- (4o);
			\draw[darkcandyapp] (2it) -- (2ot) -- (2o);
			\node (emp) at (0,-.5,.5) {$\varnothing$};
			\node (b) at (0,-.75,1) {$\{{\color{darkblue}\bullet } \}$};
			\draw[dashed] (-.56,0,.5) -- (-.5,0,.5) node[right] {$\{{\color{darkgreen}\bullet } \}$};
			\draw[dashed] (-.75,0,1) -- (-.75,0,1.5) -- (-.9,0,1.5) node[left] {$\{{\color{darkblue}\bullet },{\color{darkgreen}\bullet } \}$};
			\draw[dashed] (.75,0,1) -- (.75,0,1.5) -- (1.5,0,1.5) node[right] {$\{{\color{darkblue}\bullet },{\color{darkcandyapp}\bullet } \}$};
			\draw[dashed]  (1,0,.5) -- (1.5,0,.5) node[right] {$\{{\color{mellowyellow}\bullet },{\color{darkblue}\bullet },{\color{darkcandyapp}\bullet } \}$};
			\draw[dashed]  (0,1,.5) -- (0,1.5,.5) node[right] {$\{{\color{mellowyellow}\bullet },{\color{darkblue}\bullet },{\color{darkcandyapp}\bullet },{\color{darkgreen}\bullet} \}$};
			\node (rgb) at (0,.75,1) {$\{{\color{darkblue}\bullet } ,{\color{darkcandyapp}\bullet },{\color{darkgreen}\bullet }  \}$};
		\end{tikzpicture}
		\caption{The doubles $\mathsf{2}(\triangle)\cong \bbS^2$ and $\mathsf{2}(\square)\cong \bbT^2$, where $\triangle,\square$ denote a triangle and square, respectively, whose edges we denote ${\color{mellowyellow}\bullet },{\color{darkblue}\bullet },{\color{darkcandyapp}\bullet }, {\color{darkgreen}\bullet}$. The preimages $\pi^{-1}(\bullet)$ for $\bullet \in \{{\color{mellowyellow}\bullet },{\color{darkblue}\bullet },{\color{darkcandyapp}\bullet }, {\color{darkgreen}\bullet }\}$ have been colored in the same way as $\bullet$ itself. The various sheets of $\pi$ are labeled by the possible sets of faces, except for a handful which are hidden from view and left unlabeled to avoid cluttering the diagram.}
		\label{fig:octahedron}
	\end{center}
\end{figure}

We now discuss the topology of $\mathsf{2}(M)$, insofar as useful for providing a criterion for lifting continuous maps $\mathsf{2}(M)\to X$ to continuous maps 
\begin{equation}
\mathsf{2}(M) \to \widehat{X} 
\end{equation}
into the cover $\widehat{X} = \widetilde{X} / [\pi_1(X),\pi_1(X)]$ of $X$, whenever $X$ is a sufficiently nice topological space. This will allow us to construct the $\iota_{\ell,m,n}$ in \cref{eq:misc_14} indirectly, first by constructing a map $\mathsf{2}(A_{\ell,m,n}) \to \calM_N$ and then lifting.

Given a point $p\in M^\circ$, consider the sub-groupoid $\pi_1(\mathsf{2}(M), \Pi^{-1}(p))\subseteq \pi_1(\mathsf{2}(M))$ of homotopy classes $[\gamma:[0,1]\to \mathsf{2}(M)] \in \pi_1(\mathsf{2}(M))$ with $\{\gamma(0),\gamma(1)\}\subseteq \Pi^{-1}(p)$. This has a manageable set of generators, which is the only structural feature that we will need to know:
\begin{lemmap}
	Suppose that $M$ is contractible. For each $\mathrm{F}\in \calF(M)$, fix a continuous $\gamma_{\mathrm{F}} : [0,1] \to M$ such that $\gamma_{\mathrm{F}}(0) = p$, $\gamma_{\mathrm{F}}(1) \in \mathrm{F}^\circ$, and $\gamma_{\mathrm{F}}(t) \in M^\circ$ for all $t$ with $0<t<1$. For $\calF\subseteq \calF(M)$, let 
	\begin{equation}
		\gamma_{\calF,\mathrm{F}}(t) = 
		\begin{cases}
			[(\calF,\gamma_{\mathrm{F}}(2t)) ]_\sim & (0\leq t \leq 1/2), \\ 
			[(\calF \Delta \{\mathrm{F}\},\gamma_{\mathrm{F}}(2-2t)) ]_\sim & (1/2\leq t \leq 1).
		\end{cases} 
	\end{equation}
	Then, $\pi_1(\mathsf{2}(M), \Pi^{-1}(p))$ is generated by the classes $[\gamma_{\calF,\mathrm{F}}]$, and $[\gamma_{\calF,\mathrm{F}}]^{-1} = [\gamma_{\calF\Delta \{\mathrm{F}\},\mathrm{F}}]$. These classes do not depend on the choices of $\gamma_{\mathrm{F}}$. 
\end{lemmap}
The main part of the proof, whose details we omit, follows from the observation that any element of $\pi_1(\mathsf{2}(M), \Pi^{-1}(p))$ can be homotoped to one which avoids all $\Pi^{-1}(\mathrm{f})$ for $\mathrm{f}$ a codimension $\geq 2$ corner of $M$. In the examples depicted in \Cref{fig:octahedron}, the $\gamma_{\calF,\mathrm{F}}$ are paths going from the center of one sheet to the center of another. They are in bijection with the ``edges'' drawn.

As a corollary of the preceding lemma:
\begin{lemma}
	Let $G$ be a group. If 
	\begin{equation} 
		\Phi: \pi_1(\mathsf{2}(M), \Pi^{-1}(p)) \to G
	\end{equation} 
	is a map of groupoids such that $\Phi([\gamma_{\calF,\mathrm{F}}]) = \Phi([\gamma_{\calF',\mathrm{F}}])$ for all $\calF,\calF'\subseteq \calF(M)$ such that $\mathrm{F}\notin \calF\Delta \calF'$, then, for any $p_0\in \Pi^{-1}(p)$, the image of $\pi_1(\mathsf{2}(M),p_0) \subseteq  \pi_1(\mathsf{2}(M), \Pi^{-1}(p))$ under $\Phi$ is a subgroup of the commutator subgroup $[G,G] \subseteq G$. 
	\label{lem:lifting_lemma}
\end{lemma}
\begin{proof}
	We have $p_0 = [(\calF_0,p)]_\sim$ for some $\calF_0 \subseteq \calF(M)$. 
	By the previous lemma, any $\gamma \in \pi_1(\mathsf{2}(M),p_0)$ can be written as a well-defined composition $\gamma =  [\gamma_{\calF_0,\mathrm{F}_0}]\cdots [\gamma_{\calF_N,\mathrm{F}_N}]$ for some sequence $\mathrm{F}_j \in \calF(M)$, where, in order for the composition to be well-defined, $\calF_j = \calF_{j-1}\Delta \{\mathrm{F}_{j-1}\}$ for each $j\in \{1,\cdots,N\}$. Applying $\Phi$ to $\gamma$ yields
	\begin{equation}
		\Phi(\gamma) = \prod_{j=0}^N \Phi( [\gamma_{\calF_j,\mathrm{F}_j}]),
	\end{equation}
	We want to show that $\Phi(\gamma) \in [G,G]$, which is equivalent to showing that the image of $\Phi(\gamma)$ in the abelianization $G/[G,G]$ is trivial. 
	By assumption, $\Phi([\gamma_{\calF,\mathrm{F}}])$ depends only on $\calF$ through whether or not $\mathrm{F}\in \calF$, and by the previous lemma the two possibilities are inverses. Thus, the image of $\Phi(\gamma)$ in $G/[G,G]$ can be written 
	\begin{equation}
		\prod_{\mathrm{F}\in \calF(M)}  \Phi([\gamma_{\calF_0,\mathrm{F}}]) ^{e_{\mathrm{F}} - e_{\mathrm{F}}'}   \bmod [G,G], 
		\label{eq:misc_024} 
	\end{equation}
	where $e_{\mathrm{F}}$ is the number of $j\in \{0,\cdots,N\}$ such that $\mathrm{F}_j=\mathrm{F}$ and $\mathrm{F}\notin \calF_0 \Delta \calF_j$ and $e_{\mathrm{F}}'$ is the number of $j\in \{0,\cdots,N\}$ such that $\mathrm{F}_j=\mathrm{F}$ and $\mathrm{F}\in \calF_0 \Delta \calF_j$. 
	
	In order for $\gamma$ to end up at $p_0$, we must have $\calF_N\Delta \{\mathrm{F}_N\}=\calF_0$, which means that every $\mathrm{F}\in \calF(M)$ must appear in the sequence $\mathrm{F}_0,\mathrm{F}_1,\cdots$ an even number $e_{\mathrm{F}}''$ of times, and the definition of the $\calF_j$ implies that $e_{\mathrm{F}},e_{\mathrm{F}}'$ are both equal to $2^{-1} e_{\mathrm{F}}''$. 
	
	So, the element of $G/[G,G]$ defined by \cref{eq:misc_024} is trivial. 
\end{proof}

Given any collection $\smash{\{\varrho_{\mathrm{F}} \}_{\mathrm{F}\in \calF(M)}} \subset C^\infty(M;\bbR^+)$ of boundary-defining-functions (bdfs) of the facets $\mathrm{F}\in \calF(M)$ and a system of compatible tubular neighborhoods thereof, one can define a smooth structure on $\mathsf{2}(M)$ in a manner generalizing the case when $M$ is a manifold-with-boundary. The smooth structure is independent of the choices made only in the weak sense that the resultant $C^\infty$-manifolds are guaranteed to be diffeomorphic; the smooth structures may be incompatible.
For instance, consider $\mathsf{2}([0,\infty)_{x+x^2})$. Since $[0,\infty)_{x+x^2} = [0,\infty)_x$ at the level of $C^0$-manifolds, $\mathsf{2}([0,\infty)_{x+x^2}) = \mathsf{2}([0,\infty)_x)$ at the level of $C^0$-manifolds. The latter is naturally identified with $\bbR_x$ at the level of smooth manifolds, via the map $\bbR_x\to \mathsf{2}([0,\infty)_x)$ sending $x\mapsto [(\varnothing,x)]_\sim$ if $x\geq 0$ and $x\mapsto [(\{0\},|x|)]_\sim$ otherwise.
So, 
\begin{equation} 
\mathsf{2}([0,\infty)_{x+x^2}) \cong \bbR_x
\end{equation} 
at the level of $C^0$-manifolds, and we can consider $x:\mathsf{2}([0,\infty)_{x+x^2}) \twoheadrightarrow \bbR$. But, this map is \emph{not} a diffeomorphism. 
At the level of smooth manifolds,
\begin{equation}
	\mathsf{2}([0,\infty)_{x+x^2} ) \cong \bbR_{y },
\end{equation}
where the diffeomorphism is given by $y=x+x^2$ if $x\geq 0$ and $y=x-x^2$ if $x\leq 0$. 
This has an incompatible smooth structure with $\bbR_x$, as for instance $y\notin C^2  (\bbR_x)$ but $y\in C^\infty(\bbR_y)$.

For $\mathsf{2}(A_{\ell,m,n})$, we will give an explicit smooth structure below. 

If $M$ is orientable, then so is $\mathsf{2}(M)$. Indeed, one can choose the orientation such that, for each $\calF\subseteq \calF(M)$, the diffeomorphism $\Pi:\{[(\calF,p)]_\sim:p\in M^\circ\}\to M^\circ$ is orientation preserving if $\calF$ contains evenly many elements and orientation reversing otherwise. This stipulation defines an orientation on $\mathsf{2}(M)\backslash \Pi^{-1}(\partial M)$, so the claim is that this orientation can be extended to $\Pi^{-1}(\partial M)$. 
Working in local coordinates, it can be seen that the orientations on the components of this dense submanifold are compatible, so the orientations can be extended and stitched together. For instance, pulling back the volume form $\mathrm{d} t_1\wedge \cdots \wedge \mathrm{d}t_N$ on $(0,\infty)^N_t$ via the map $\Pi_{\mathrm{Model}}:\bbR^N_\tau\to [0,\infty)^N_t$ given by $(\tau_1,\cdots,\tau_N)\mapsto (|\tau_1|,\cdots,|\tau_N|)$,
the result (away from the hyperplanes $\{\tau_j=0\}$, at each of which the pullback is undefined) is the form $\pm \mathrm{d} \tau_1\wedge \cdots \wedge  \mathrm{d} \tau_N$, where the sign is positive if an even number of the $\tau$'s are negative and negative otherwise. Therefore, $\Pi_{\mathrm{Model}}$ is orientation preserving on half of the $1/2^N$th-ants of $\bbR_\tau^N$ and orientation reversing on the rest. The case of general orientable $M$ can be reduced to this example via passage to local coordinate charts, the map $\Pi_{\mathrm{Model}}$ being a local model for $\Pi$.

\section{The structure of $\mathsf{2}(A_{\ell,m,n})$}

In this section, we discuss the structures of $A_{\ell,m,n}$ and of the double $\mathsf{2}(A_{\ell,m,n})$. In \cite{Sussman}, the $A_{\ell,m,n}$ were constructed by blowing up various boundary components (of various dimensions) of the $N=\ell+m+n$-cube $\square^{\ell,m,n} \cong \square^N$. This picture is worth keeping in mind, in part because it allows us to label the boundary hypersurfaces of $A_{\ell,m,n}$ by particular boundary components of the cube, the nonempty subsets of the form
\begin{equation}
\mathrm{f}_{\calS;x_0} = \{ (x_1,\ldots,x_N)\in \square^{\ell,m,n} : x_j = x_0\text{ whenever } j\in \calS \} \subsetneq \square^{\ell,m,n}
\end{equation}
for $x_0 \in \{0,1,\infty\}$ and $\varnothing\subsetneq \calS\subset \{1,\ldots,N\}$, with the caveat that we consider $-\infty=\infty$. Then, $A_{\ell,m,n}$ is formed by blowing up the $\mathrm{f}_{\calS;x_0}$ in order of increasing dimension. Then, the various boundary hypersurfaces of $A_{\ell,m,n}$ are denoted $\mathrm{F}_{\calS;x_0} \subset A_{\ell,m,n}$, with $\mathrm{F}_{\calS;x_0}$ being the subset of $A_{\ell,m,n}$ blowing down to $\mathrm{f}_{\calS;x_0}$.

In other words, letting $\calI_1=\{1,\cdots,\ell\}$, $\calI_2 = \{\ell+1,\cdots,\ell+m\}$, and $\calI_3=\{\ell+m+1,\cdots,N\}$.  the boundary hypersurfaces of $A_{\ell,m,n}$ are all of precisely one of the following three forms: $\mathrm{F}_{\calS;0}$ for $\calS\subseteq \calI_1 \cup \calI_2$, $\mathrm{F}_{\calS;1}$ for $\calS\subseteq\calI_2\cup \calI_3$, or $\mathrm{F}_{\calS;\infty}$ for $\calS\subseteq \calI_1\cup \calI_3$. The casework seen here will be repeated below, in various forms. Its origin is that the Dotsenko--Fateev integrands are typically singular when any of the $x_1,\ldots,x_N$ are $0$, $1$, or $\infty$, and each of these three cases is treated separately. 

We now exhibit an explicit atlas discussed in \cite[Appendix B]{Sussman}.
This atlas can be written in terms of functions $a_1,\cdots,a_N:\square^{\ell,m,n}_x\to [0,1]$  defined as follows:
for each $j\in \calI_1$, let $a_j = 1/(1-x_j)$, for each $j\in \calI_3$, let $a_j = (x_j-1)/x_j$, and for $j\in \calI_2$, just let $a_j = x_j$. Then, $(x_1,\cdots,x_N)\mapsto (a_1,\cdots,a_N)$ is a diffeomorphism 
\begin{equation} 
\square^{\ell,m,n}_x \cong \square_{a}^N.
\end{equation} 
For each triple $(\calS_1,\calS_2,\calS_3)$ of subsets $\calS_1\subseteq \calI_1$, $\calS_2\subseteq \calI_2$, and $\calS_3\subseteq \calI_3$, let 
\begin{equation}
	\xi_j[\calS_1,\calS_2,\calS_3] = 
	\begin{cases}
		a_j & (j\in \calS_1\cup \calS_2\cup \calS_3), \\ 
		1-a_j & (\text{otherwise}).
	\end{cases}
\end{equation}
We have $\square_{a}^N\cong \square_{\xi}^N$, and thus $\square^{\ell,m,n}_x\cong \square_\xi^N$. These coordinates are useful in describing the atlas near the lift in $A_{\ell,m,n}$ of the corner 
\begin{equation} 
\{a_j =0\text{ for }j\in \calS\text{ and }a_j=1\text{ for }j\in \calS^\complement\} \subset \square^N_a,
\end{equation} 
where $\calS=\calS_1\cup \calS_2\cup \calS_3$ and $\calS^\complement = \{1,\ldots,N\}\backslash \calS$. We come finally to coordinates used in the actual atlas: for each triple $(\sigma_\infty,\sigma_0,\sigma_1)$ of permutations 
\begin{align}
	\begin{split} 
		\sigma_\infty &: \calS_1 \cup \calS_3^\complement\to \calS_1\cup \calS_3^\complement,  \\ 
		\sigma_0 &: \calS_2 \cup \calS_1^\complement\to \calS_2\cup \calS_1^\complement, \\ 
		\sigma_1 &:  \calS_3\cup \calS_2^\complement \to \calS_3 \cup \calS_2^\complement, 
	\end{split} 
\end{align}
where $\calS_1^\complement = \calI_1\backslash \calS_1$, $\calS_2^\complement= \calI_2\backslash \calS_2$, and $\calS_3^\complement= \calI_3\backslash \calS_3$, consider, for each $j\in\{1,\cdots,N\}$, 
\begin{equation}
	\xi_j[\calS_1,\calS_2,\calS_3;\sigma_\infty,\sigma_0,\sigma_1](t_1,\cdots,t_N) = 
	\begin{cases}
		\prod_{k \in \calS_1\cup \calS_3^\complement , k\leq \sigma_\infty(j)} t_{k} & (j\in \calS_1\cup \calS_3^\complement), \\
		\prod_{k \in \calS_2\cup \calS_1^\complement, k\leq \sigma_0(j)}\; t_{k} & (j\in \calS_2\cup \calS_1^\complement), \\
		\prod_{k\in \calS_3\cup \calS_2^\complement, k\leq \sigma_1(j)}\; t_{k} & (j\in \calS_3\cup \calS_2^\complement).
	\end{cases}
\end{equation}
We explain how these yield an atlas.
Consider the map
\begin{align} 
	\begin{split} 
		\xi[\calS_1,\calS_2,\calS_3;\sigma_\infty,\sigma_0,\sigma_1]:&(0,\infty)^N_t\to (0,\infty)^N_\xi \\
		:&(t_1,\cdots,t_N) \mapsto (\xi_1[\calS_1,\calS_2,\calS_3;\sigma_\infty,\sigma_0,\sigma_1],\cdots,\xi_N[\calS_1,\calS_2,\calS_3;\sigma_\infty,\sigma_0,\sigma_1]).
	\end{split} 
	\label{eq:misc_02b}
\end{align} 
Inverting this, the  $t_1,\cdots,t_N$ are ratios of $\xi_1[\calS_1,\calS_2,\calS_3;\sigma_\infty,\sigma_0,\sigma_1],\cdots,\xi_N[\calS_1,\calS_2,\calS_3;\sigma_\infty,\sigma_0,\sigma_1]$. Let $U[\calS_1,\calS_2,\calS_2;\sigma_\infty,\sigma_0,\sigma_1]\subseteq [0,\infty)^N$ 
denote the relatively open subset of $[0,\infty)^N$ consisting of $(t_1,\cdots,t_N)\in [0,\infty)^N$ such that 
\begin{equation} 
	\xi_j[\calS_1,\calS_2,\calS_3;\sigma_\infty,\sigma_0,\sigma_1] < 1 
\end{equation}
for all $j\in \{1,\cdots,N\}$. 
Composing the map \cref{eq:misc_02b} with the diffeomorphism $\smash{\square_\xi^N \to \square_x^{\ell,m,n}}$, we get a smooth map $U[\calS_1,\calS_2,\calS_2;\sigma_\infty,\sigma_0,\sigma_1] \to \square^{\ell,m,n}_x$.
The key claim, the reason for which is given in \cite{Sussman}, is that this lifts via the blowdown map $\mathrm{bd}:A_{\ell,m,n}\to \square^{\ell,m,n}_x$ to an embedding
\begin{equation}
	X[\calS_1,\calS_2,\calS_3;\sigma_\infty,\sigma_0,\sigma_1]:U[\calS_1,\calS_2,\calS_2;\sigma_\infty,\sigma_0,\sigma_1] \to A_{\ell,m,n}. 
	\label{eq:misc_032}
\end{equation}
The collection $\calU$ of all the sets $U[\calS_1,\calS_2,\calS_3;\sigma_\infty,\sigma_0,\sigma_1]$ of the form above, as $\calS_1,\calS_2,\calS_3,\sigma_\infty,\sigma_0,\sigma_1$ vary, cover $A_{\ell,m,n}$. 
To see that the transition maps are smooth, and therefore that the maps above give a smooth atlas, it suffices to note that the transition maps 
\begin{equation}
\bbR^N_{t_1,\ldots,t_N} \supseteq U[\calS_1,\calS_2,\calS_2;\sigma_\infty,\sigma_0,\sigma_1] \to \bbR^N
\end{equation}
have components which are rational functions of $t_1,\cdots,t_N$, with the denominators nonvanishing on the relevant domains. So, the maps \cref{eq:misc_032} give a smooth atlas.

\begin{figure}
	\begin{center}
		\begin{tikzpicture}[scale = 2]
			\coordinate (origin) at (0,0,0);
			\coordinate (ul) at (0,1,.5);
			\coordinate (ur) at (.5,1,0);
			\coordinate (ru) at (1,.5,0);
			\coordinate (rl) at (1,0,.5);
			\coordinate (ll) at (.5,0,1);
			\coordinate (lu) at (0,.5,1);
			\coordinate (z) at (0,1,0);
			\coordinate (x) at (0,0,1);
			\coordinate (y) at (1,0,0);
			\filldraw[lightgray!10] (ul) -- (0,2,.5) -- (0,2,2) -- (0,.5,2.5) -- (lu) -- cycle;
			\filldraw[lightgray!10] (ur) -- (.5,2,0) -- (1.15,2,0) -- (2,.5,0) -- (ru) -- cycle;
			\filldraw[lightgray!10] (rl) -- (2,0,.5) -- (2,0,2) -- (.5,0,2.5) -- (ll) -- cycle;
			\filldraw[lightgray!10] (ul) -- (0,2,.5) -- (.5,2,0) -- (ur) -- cycle;
			\filldraw[lightgray!10] (ru) -- (2,.5,0) --  (2,0,.5) -- (rl) -- cycle;
			\filldraw[lightgray!10] (ll) -- (.5,0,2.5) --  (0,.5,2.5) -- (lu) -- cycle;
			\filldraw[fill=lightgray!10] (ul) -- (ur) -- (ru) -- (rl) -- (ll) -- (lu) -- cycle;
			\draw[->] (ul) -- (0,2,.5);
			\draw[->] (ur) -- (.5,2,0);
			\draw[->] (ru) -- (2,.5,0);
			\draw[->] (rl) -- (2,0,.5);
			\draw[->] (ll) -- (.5,0,2.5);
			\draw[->] (lu) -- (0,.5,2.5);
			\draw[darkcandyapp, ->] (0,1,.8) -- (0,1.4,.8) node[above left] {$z$};
			\draw[darkcandyapp, ->] (0,1,.8) -- (0,.8,1.025) node[left] {$x/z$};
			\draw[darkcandyapp, ->] (0,1,.8) -- (.25,1,.5) node[below] {$y/x$};
			\draw[darkblue, ->] (.6,1.1,0) -- (.6,1.4,0) node[right] {$z$}; 
			\draw[darkblue, ->] (.6,1.1,0) -- (.8,.9,0) node[above right] {$y/z$};
			\draw[darkblue, ->] (.6,1.1,0) -- (.4,1.15,.35) node[above] {$x/y$};
			\draw[darkgreen, ->] (1,.6,0) -- (1.3,.6,0) node[above right] {$y$};
			\draw[darkgreen, ->] (1,.6,0) -- (.825,.8,0) node[right] {$z/y$};
			\draw[darkgreen, ->] (1,.6,0) -- (1,.4,.2) node[above left] {$x/z$};
			\draw[darkgreen, ->] (.1,.5,1) -- (.1,.5,1.5) node[below] {$x$};
			\draw[darkgreen, ->] (.1,.5,1) -- (.1,.7,.825) node[right] {$z/x$};
			\draw[darkgreen, ->] (.1,.5,1) -- (.3,.3,1) node[above right] {$y/z$};
			\draw[darkcandyapp, ->] (1,0,.7) -- (1.3,0,.7) node[below] {$y$};
			\draw[darkcandyapp, ->] (1,0,.7) -- (.85,0,.85) node[below right] {$x/y$};
			\draw[darkcandyapp, ->] (1,0,.7) -- (1,.25,.45) node[left] {$z/x$};
			\draw[darkblue, ->] (.6,0,1.2) -- (.6,0,1.7) node[below] {$x$};
			\draw[darkblue, ->] (.6,0,1.2) -- (.75,0,1.05) node[below] {$y/x$};
			\draw[darkblue, ->] (.6,0,1.2) -- (.4,0.2,1.2) node[below left] {$z/y$};
		\end{tikzpicture}
	\end{center}
	\caption{The $|\frakS_3|=6$ coordinate systems $X[\varnothing,\{1,2,3\},\varnothing ; 1,\sigma,1]^{-1}$ near the $\{x_1,x_2,x_3=0\}$ corner of $A_{0,3,0}$ (say that hidden from view in \Cref{fig:A}). Here, $x=x_1$, $y=x_2$, and $z=x_3$, defined on $\smash{A_{0,3,0}^\circ} = (\square^{3})^\circ$. The origins of the coordinate systems have been depicted slightly shifted to aid readability. In each coordinate system, each of $t_1,t_2,t_3$ is a ratio of two of $1,x,y,z$.}
	\label{fig:coords}
\end{figure}
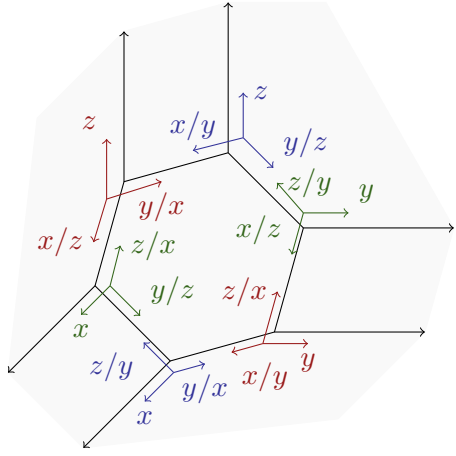

We have phrased this as a theorem, given the definition of $A_{\ell,m,n}$ via iterated blowups, but it can also be taken as a \emph{definition} of the $A_{\ell,m,n}$. Since this avoids the need to use the argument in \cite[Appendix B]{Sussman}, and since we make no use of the blowup construction below (only the blowdown map $\mathrm{bd} : A_{\ell,m,n} \to \square^{\ell,m,n}_x$, which can be defined explicitly in terms of coordinate charts) this is actually more convenient for our purposes. 

See \Cref{fig:coords} for an illustration in the $\ell=0,m=3,n=0$ case, focusing on a neighborhood of $\{x_1,x_2,x_3=0\}$, so that $\xi_j = x_j$. Consider the identity permutation $\sigma_0 = 1$. Then, the 
\begin{equation} 
x_j=\xi_j[\varnothing,\{1,2,3\},\varnothing; 1,1,1] =\xi_j[\calS_1,\calS_2,\calS_3;\sigma_\infty,\sigma_0,\sigma_1]
\end{equation} 
are given by $x_1 = t_1$, $x_2 = t_1 t_2$, and $x_3 = t_1t_2t_3$. In other words, $t_1=x_1$, $t_2 = x_2/x_1$, and $t_3 = x_3/x_2$ serve as local coordinates. Replacing $\sigma_0$ with another permutation results in permuting $x_1,x_2,x_3$ in the formulas for $t_1,t_2,t_3$. 

In the coordinate system $X[\calS_1,\calS_2,\calS_3;\sigma_\infty,\sigma_0,\sigma_1]$, the subset $\{t_k=0\} \subset [0,\infty)^N_t$ is mapped to a subset of a unique boundary hypersurface of $A_{\ell,m,n}$, namely 
\begin{equation} 
	\mathrm{F}[\calS_1,\calS_2,\calS_3;\sigma_\infty,\sigma_0,\sigma_1;k] =  \mathrm{F}_{\calS; x_0} \in \calF(A_{\ell,m,n}),
\end{equation} 
where $x_0 = \infty$ if $k\in \calS_1\cup \calS_3^\complement$, $x_0=0$ if $k\in \calS_2\cup \calS_1^\complement$, and $x_0=1$ if $k\in \calS_3\cup \calS_2^\complement$, and where $\calS$ is the set of all $j$ in the domain of $\sigma_{x_0}$  such that $k\leq \sigma_{x_0}(j)$. For instance, in the coordinate system $\xi_j[\varnothing,\{1,2,3\},\varnothing; 1,1,1]$ worked out above, $\{t_1=0\} = \{x_1,x_2,x_3=0\}$, $\{t_2 = 0\} = \{x_2,x_3=0\}$, and $\{t_3 = 0 \} = \{x_3=0\}$.

In the same way that the coordinates $t_1,\cdots,t_N$ on $[0,\infty)^N_{t_1,\cdots,t_N}$ extend to coordinates on 
\begin{equation} 
\mathsf{2}([0,\infty)^N_{t_1,\cdots,t_N})\cong \bbR_{t_1,\cdots,t_N}^N,
\end{equation} 
the coordinate systems above can be extended to coordinate systems on $\mathsf{2}(A_{\ell,m,n})$. Consider the open subset 
\begin{equation} 
	V[\calS_1,\calS_2,\calS_3;\sigma_\infty,\sigma_0,\sigma_1] = \{(t_1,\cdots,t_N)\in \bbR^N : (|t_1|,\cdots,|t_N|) \in U[\calS_1,\calS_2,\calS_3;\sigma_\infty,\sigma_0,\sigma_1]\}.
\end{equation} 
Then, for each subset $\calF\subseteq \calF(A_{\ell,m,n})$, consider 
\begin{equation} 
	\widehat{X}[\calS_1,\calS_2,\calS_3;\sigma_\infty,\sigma_0,\sigma_1;\calF]: 	V[\calS_1,\calS_2,\calS_3;\sigma_\infty,\sigma_0,\sigma_1]  \to \mathsf{2}(A_{\ell,m,n})
	\label{eq:Psi_coord}
\end{equation} 
defined by 
\begin{equation}
	(t_1,\cdots,t_N) \mapsto [(\calF\Delta \calG(t_1,\cdots,t_N), X[\calS_1,\calS_2,\calS_3;\sigma_\infty,\sigma_0,\sigma_1](|t_1|,\cdots,|t_N|))]_\sim, 
\end{equation}
where 
\begin{equation}
	\calG(t_1,\cdots,t_N) = \{\mathrm{F}[\calS_1,\calS_2,\calS_3;\sigma_\infty,\sigma_0,\sigma_1;k] : t_k<0\}.
\end{equation}
These define homeomorphisms onto their images and serve as a $C^0$-atlas for $\mathsf{2}(A_{\ell,m,n})$. Moreover, they serve as a $C^\infty$-atlas. In order to see this, it must be checked that the various transition maps are smooth. Like the transition maps for the atlas for $A_{\ell,m,n}$ discussed above, these transition maps are rational functions of the $t_1,\cdots,t_N$ with the denominators are nonvanishing on the relevant domains. They are therefore smooth.

Important submanifolds of $A_{\ell,m,n}$ are
\begin{equation}
	H_{j,k}  = \mathrm{cl}_{A_{\ell,m,n}} \{ (x_1,\cdots,x_N) \in \square^{\ell,m,n\circ }_x : x_j=x_j \}
\end{equation}
for distinct $j,k\in \calI_1$, distinct $j,k\in \calI_2$, or distinct $j,k\in \calI_3$. Their importance stems from the fact that they are the loci of the singularities of the Dotsenko--Fateev integrand $\omega$ remaining \emph{in the interior} of $A_{\ell,m,n}$. The key idea exploited here, besides those already in \cite{KT1, KT2}, is handling these singularities using a different technique than the singularities on the boundary.  

The conceptual reason why it is easier to analyze the interior singularities on $A_{\ell,m,n}$ than on $\smash{\square^{\ell,m,n}}$ is because the $H_{j,k}$ are \emph{interior p-submanifolds}, a concept due to Melrose. Roughly, this means that they locally look like subsets of $[0,\infty)^K \times \bbR^J$ of the form $[0,\infty)^K \times \bbR^{J-\kappa}\times \{0\}^{\kappa}$, for some $K,J,\kappa$, and in this case $\kappa=1$. The claim can be seen in the $N=3$ case in \Cref{fig:interior_singularities}. The ur-example of a submanifold that fails to be a p-submanifold is the diagonal $\{x_1=x_2\} \subset \square^2$, or more generally the blowdowns of the $H_{j,k}$ in the $\square^{\ell,m,n}$. To be clear, this conceptual point-of-view is not necessary below, but it explain why the analysis goes through.

Consider the $C^0$-submanifolds $\Pi^{-1}(H_{j,k})$ of $\mathsf{2}(A_{\ell,m,n})$. These are smooth submanifolds. Since $\Pi$ is a covering map of $A_{\ell,m,n}^\circ$, this only needs to be checked near $\Pi^{-1}(\partial H_{j,k})$. One can use the local coordinate system above: given $j,k$ either both in $\calS_1\cup \calS_2\cup \calS_3$ or both not, then 
\begin{equation}
	\Pi^{-1}(H_{j,k}) = \widehat{X}[\calS_1,\calS_2,\calS_3;\sigma_\infty,\sigma_0,\sigma_1;\calF]^{-1}\Big(\Big\{ (t_1,\cdots,t_N)\in V : \prod_{s\in \calS} t_s^2 = 1\Big\} \Big)
\end{equation}
locally, for some subset $\calS \subseteq \{1,\cdots,N\}$ depending on $\sigma_\infty,\sigma_1,\sigma_0$ and on $j,k$. The right-hand side defines a smooth submanifold of $\bbR^N_t$. These coordinate charts cover $\Pi^{-1}(\partial H_{j,k})$.

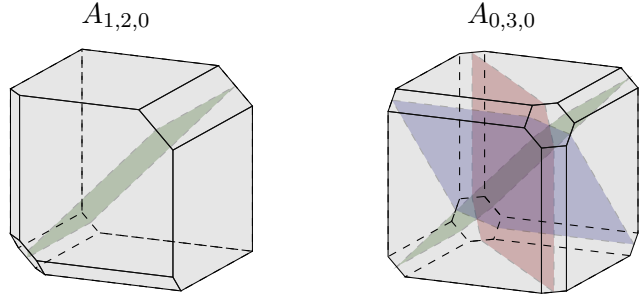
\begin{figure}[t]
	\begin{center} 
		\tdplotsetmaincoords{75}{115}
		\begin{tikzpicture}[scale=2.5,tdplot_main_coords]
			\draw[opacity=0] (1.5,0,0) -- (0,1.5,0) -- (0,0,1.5); 
			\draw[fill=darkgreen, opacity=.3, dashed] (0,.05,.05) -- (0,.9,.9) -- (1,.9,.9) -- (1,.15,.15) -- (.8,.05,.05) -- cycle;%
			\draw[fill=gray,fill opacity = .1] (1,.2,.1) -- (1,.1,.2) -- (1,.1,1) -- (1,.8,1) -- (1,1,.8) -- (1,1,.1) -- cycle;%
			\draw[fill=gray,fill opacity = .1] (1,1,.1) -- (.9,1,0) -- (0,1,0) -- (0,1,.8) -- (1,1,.8) -- cycle;%
			\draw[fill=gray,fill opacity = .1] (1,.8,1) -- (0,.8,1) -- (0,0,1) -- (.9,0,1) -- (1,.1,1) -- cycle;%
			\draw[fill=gray,fill opacity = .1] (1,.8,1) -- (0,.8,1) -- (0,1,.8) -- (1,1,.8) -- cycle;%
			\draw[fill=gray,fill opacity = .1] (1,1,.1) -- (.9,1,0) -- (.9,.2,0) -- (1,.2,.1) -- cycle;%
			\draw[fill=gray,fill opacity = .1] (1,.1,1) -- (.9,0,1) -- (.9,0,.2) -- (1,.1,.2) -- cycle;%
			\draw[fill=gray,fill opacity = .1] (1,.2,.1) -- (1,.1,.2) -- (.9,0,.2) -- (.8,0,.1) -- (.8,.1,0) -- (.9,.2,0) -- cycle;%
			\draw[fill=gray,fill opacity = .1, draw=none] (0,.1,0) -- (0,0,.1) -- (.8,0,.1) -- (.8,.1,0) -- cycle;
			\draw[fill=gray,fill opacity = .1, dashed] (.9,.2,0) -- (.9,1,0) -- (0,1,0) -- (0,.1,0) -- (.8,.1,0) -- cycle;%
			\draw[fill=gray,fill opacity = .1, dashed] (.8,0,.1) -- (.9,0,.2) -- (.9,0,1) -- (0,0,1) -- (0,0,.1) -- cycle;%
			\draw[fill=gray,fill opacity = .1, dashed] (0,0,.1) -- (0,.1,0) -- (0,1,0) -- (0,1,.8) -- (0,.8,1) -- (0,0,1) -- cycle;%
			\node (lb) at (0,.2,1.2) {$A_{1,2,0}$};
		\end{tikzpicture}
		\begin{tikzpicture}[scale=2.5,tdplot_main_coords]
			\draw[opacity=0] (1.5,0,0) -- (0,1.5,0) -- (0,0,1.5); 
			\draw[fill=darkgreen, opacity=.3, dashed] (1,.85,.85) -- (.8,.95,.95) -- (0,.95,.95) -- (0,.15,.15) -- (.2,.05,.05)-- (1,.05,.05) -- cycle;%
			\draw[fill=darkcandyapp, opacity=.3, dashed] (.85,.85,1) -- (.95,.95,.8) -- (.95,.95,0) -- (.15,.15,0) -- (.05,.05,.2)-- (.05,.05,1) -- cycle;%
			\draw[fill=darkblue, opacity=.3, dashed] (.85,1,.85) -- (.95,.8,.95) -- (.95,0,.95) -- (.15,0,.15) -- (.05,.2,.05)-- (.05,1,.05) -- cycle;%
			\draw[fill=gray,fill opacity = .1] (1,.9,.8) -- (1,.8,.9) -- (.9,.8,1) -- (.8,.9,1) -- (.8,1,.9) -- (.9,1,.8) -- cycle;%
			\draw[fill=gray,fill opacity = .1] (1,.9,.8) -- (.9,1,.8) -- (.9,1,0) -- (1,.9,0) -- cycle;%
			\draw[fill=gray,fill opacity = .1] (1,.8,.9) -- (.9,.8,1) -- (.9,0,1) -- (1,0,.9) -- cycle;%
			\draw[fill=gray,fill opacity = .1] (.8,.9,1) -- (.8,1,.9) -- (0,1,.9) -- (0,.9,1) -- cycle;%
			\draw[fill=gray,fill opacity = .1] (.8,1,.9) -- (.9,1,.8) -- (.9,1,0) -- (.1,1,0) -- (0,1,.1) -- (0,1,.9) -- cycle;%
			\draw[fill=gray,fill opacity = .1] (1,.9,.8) -- (1,.8,.9) -- (1,0,.9) -- (1,0,.1) -- (1,.1,0) -- (1,.9,0) -- cycle;%
			\draw[fill=gray,fill opacity = .1] (.9,.8,1) -- (.8,.9,1) -- (0,.9,1) -- (0,.1,1) -- (.1,0,1) -- (.9,0,1) -- cycle;%
			\draw[fill=gray,fill opacity = .1,dashed] (0,.1,.2) -- (0,.2,.1) -- (.1,.2,0) -- (.2,.1,0) -- (.2,0,.1) -- (.1,0,.2) -- cycle;%
			\draw[fill=gray,fill opacity = .1,draw=none] (0,.1,.2) -- (.1,0,.2) -- (.1,0,1) -- (0,.1,1) -- cycle;%
			\draw[fill=gray,fill opacity = .1,draw=none] (0,.2,.1) -- (.1,.2,0) -- (.1,1,0) -- (0,1,.1) -- cycle;%
			\draw[fill=gray,fill opacity = .1,draw=none] (.2,.1,0) -- (.2,0,.1) -- (1,0,.1) -- (1,.1,0) -- cycle;%
			\draw[fill=gray,fill opacity = .1,dashed] (.2,0,.1) -- (.1,0,.2) -- (.1,0,1) -- (.9,0,1) -- (1,0,.9) -- (1,0,.1) -- cycle;%
			\draw[fill=gray,fill opacity = .1,dashed] (0,.1,.2) -- (0,.2,.1) -- (0,1,.1) -- (0,1,.9) -- (0,.9,1) -- (0,.1,1) -- cycle;%
			\draw[fill=gray,fill opacity = .1,dashed] (.1,.2,0) -- (.2,.1,0) -- (1,.1,0) -- (1,.9,0) -- (.9,1,0) -- (.1,1,0) -- cycle;%
			\node (lb) at (0,.2,1.2) {$A_{0,3,0}$}; %
		\end{tikzpicture}
	\end{center}
	\caption{The sets $H_{j,k}$ in $A_{1,2,0}$ and $A_{0,3,0}$. Pictured are $\color{darkcandyapp} H_{1,2}$, $\color{darkblue} H_{1,3}$, and $\color{darkgreen} H_{2,3}$. }
	\label{fig:interior_singularities}
\end{figure}


For each $\epsilon \in (0,1/3)$, let $\psi_\epsilon \in C^\infty(\bbR^{\geq 0})$ be some monotonic function such that $\psi_\epsilon(t) = 1-2\epsilon$ if $t\geq 1-\epsilon$ and $\psi_\epsilon(t) = t$ if $t \leq  1-3\epsilon$ and such that 
\begin{equation} 
	\psi_\epsilon(t)\leq t
	\label{eq:misc_975}
\end{equation} 
for all $t\geq 0$. 
Consider the functions $\varrho_{\mathrm{F},\epsilon}\in C^\infty((\square_x^{\ell,m,n})^\circ;\bbR^{\geq 0})$ given by 
\begin{align}
	\begin{split} 
		\varrho_{\mathrm{F}_{\calS;0} ,\epsilon} &= \Big[\prod_{\calS\subseteq \calS_0 \subseteq \calI_1 \cup \calI_2} \Big[ \sum_{j\in \calS_0 \cap \calI_1} \psi_\epsilon(1-a_j)^2 + \sum_{j\in \calS_0\cap \calI_2} \psi_\epsilon(a_j)^2   \Big]^{(-1)^{|\calS|- |\calS_0|} } \Big]^{1/2},  \\
		\varrho_{\mathrm{F}_{\calS;1},\epsilon } &= \Big[\prod_{\calS\subseteq \calS_0 \subseteq \calI_2 \cup \calI_3} \Big[ \sum_{j\in \calS_0 \cap \calI_2} \psi_\epsilon(1-a_j)^2 + \sum_{j\in \calS_0\cap \calI_3} \psi_\epsilon (a_j)^2   \Big]^{(-1)^{|\calS|- |\calS_0|} } \Big]^{1/2},  \\
		\varrho_{\mathrm{F}_{\calS;\infty},\epsilon } &= \Big[\prod_{\calS\subseteq \calS_0 \subseteq \calI_3\cup \calI_1} \Big[ \sum_{j\in \calS_0 \cap \calI_3} \psi_\epsilon(1-a_j)^2 + \sum_{j\in \calS_0\cap \calI_1} \psi_\epsilon(a_j)^2   \Big]^{(-1)^{|\calS|- |\calS_0|} }\Big]^{1/2}.
	\end{split} 
	\label{eq:misc_029}
\end{align}
As can be seen in local coordinates, $\varrho_{\mathrm{F},\epsilon} \in C^\infty(A_{\ell,m,n}; \bbR^{\geq 0})$ for each $\mathrm{F}\in \calF(A_{\ell,m,n})$. The reason why $\psi_\epsilon$ is required for smoothness is that 
\begin{equation} 
1-|t| \notin C^\infty(\bbR_t),
\end{equation} 
but $\psi_\epsilon(1-|t|) \in C^\infty[-1,+1]$. 
As suggested by the notation, $\varrho_{\mathrm{F},\epsilon}$ is a bdf for $\mathrm{F}$. 
This can be seen directly in local coordinates, but it is also possible to prove using the inductive construction of $A_{\ell,m,n}$ via blowups. The argument is that in \cite[Appendix B]{Sussman}, for which it suffices to note that $\psi_\epsilon(a_j)$ is a bdf of the boundary hypersurface of $\square^{\ell,m,n}_x$ at which it vanishes.

If the reader is only interested in $C^0$-structure, then $\epsilon$ can be taken $\to 0^+$ without ill effects. In this limit, $\psi_\epsilon(t)$ is replaced by $t$ whenever $t\leq 1$. This is the case when plugging in $t= 1-a_j$ and $t=a_j$, as in \cref{eq:misc_029}.

The $\varrho_{\mathrm{F},\epsilon}$ satisfy several useful algebraic relations, among which are 
\begin{equation}
	\psi_\epsilon(a_j) = 
	\begin{cases}
		\prod_{j\in \calS\subseteq \calI_1\cup \calI_3} \varrho_{\mathrm{F}_{\calS;\infty},\epsilon} & (j\in \calI_1), \\ 
		\prod_{j\in \calS\subseteq \calI_2\cup \calI_1} \varrho_{\mathrm{F}_{\calS;0},\epsilon} & (j\in \calI_2), \\
		\prod_{j\in \calS\subseteq \calI_3\cup \calI_2} \varrho_{\mathrm{F}_{\calS;1},\epsilon} & (j\in \calI_3), 
	\end{cases} 
	\quad 
	\psi_\epsilon(1-a_j) = 
	\begin{cases}
		\prod_{j\in \calS\subseteq \calI_2\cup \calI_1} \varrho_{\mathrm{F}_{\calS;0},\epsilon} & (j\in \calI_1), \\ 
		\prod_{j\in \calS\subseteq \calI_3\cup \calI_2} \varrho_{\mathrm{F}_{\calS;1},\epsilon} & (j\in \calI_2), \\
		\prod_{j\in \calS\subseteq \calI_1\cup \calI_3} \varrho_{\mathrm{F}_{\calS;\infty},\epsilon} & (j\in \calI_3), 
	\end{cases}
	\label{eq:product_identity}
\end{equation}
which give a simple algebraic way of recovering $\psi_\epsilon(a_j),\psi_\epsilon(1-a_j)$ given the values of the $\varrho_{\mathrm{F},\epsilon}$.

Extend each $\varrho_{\mathrm{F},\epsilon}$ to a map $\tilde{\varrho}_{\mathrm{F},\epsilon}: \mathsf{2}(A_{\ell,m,n}) \to \bbR$ by 
\begin{equation}
	\tilde{\varrho}_{\mathrm{F},\epsilon}([(\calF,p)]_\sim ) = 
	\begin{cases}
		\varrho_{\mathrm{F},\epsilon}(p) & (\mathrm{F}\notin \calF), \\
		-\varrho_{\mathrm{F},\epsilon}(p) & (\mathrm{F} \in \calF). 
	\end{cases}
	\label{eq:misc_h64}
\end{equation}
This is certainly a well-defined continuous function on $\mathsf{2}(A_{\ell,m,n})$. Well-definedness means that the right-hand side does not depend on the element of $[(\calF,p)]_\sim$ picked. But, $[(\calF,p)]_\sim$ is a singleton unless $\varrho_{\mathrm{F},\epsilon}(p)=0$, in which case 
\begin{equation}
\varrho_{\mathrm{F},\epsilon}(p) = -\varrho_{\mathrm{F},\epsilon}(p).
\end{equation}
Less obvious, but not necessary for the proof of the main theorem, is the fact that \cref{eq:misc_h64} defines a smooth function on $\mathsf{2}(A_{\ell,m,n})$. As elsewhere in this section, this can be checked in the given local coordinate charts.

\section{Main construction}

The following functions are the basic building blocks of our multi-contours: 
for each $\delta>0$ and $\varepsilon\in (0,\delta/2)$, let $P_{\delta,\varepsilon}\in C^\infty(\bbR;\bbC)$ be defined by 
\begin{equation}
	P_{\delta,\varepsilon}(r) = \Theta_{\mathrm{reg}}\Big(\frac{|r|-\delta}{\varepsilon} \Big) |r| + \Big(1- \Theta_{\mathrm{reg}}\Big(\frac{|r|-\delta}{\varepsilon} \Big) \Big) \delta e^{\pi i(1-r/\delta)}, 
\end{equation}
where $\Theta_{\mathrm{reg}}\in C^\infty(\bbR;[0,1])$ is any mollified version of the Heaviside function $\Theta$ such that $\Theta_{\mathrm{reg}}(t)=0$ for $t\leq -2$ and $\Theta_{\mathrm{reg}}(t)=1$ for $t\geq 0$. 
For instance, we can take 
\begin{equation}
	\Theta_{\mathrm{reg}}(t) = \Big(\int_{-1}^{+1} e^{-1/(1-s^2)} \dd s \Big)^{-1}\int_{-1}^{\min\{t+1,1\}} e^{-1/(1-s^2)} \dd s. 
	\label{eq:misc_040}
\end{equation}
We have $\lim_{\varepsilon\to 0^+} P_{\delta,\varepsilon}(r)=P_{\delta}(r)$, where 
\begin{equation}
	P_{\delta}(r) = 
	\begin{cases}
		|r| & (|r|\geq \delta), \\ 
		\delta e^{\pi i (1 - r/\delta)} & (|r|\leq\delta).
	\end{cases}
\end{equation}
The ''$P$'' stands for Pochhammer. 

Note that $P_{\delta,\varepsilon}$ is nonvanishing, and as $r$ decreases from $\delta$ to $-\delta$, $P_{\delta,\varepsilon}(r)$ winds around the origin once counter-clockwise.

Since $P_{\delta,\varepsilon}(r)$ is a convex combination of complex numbers of magnitude $|r|$ and $\delta$, 
\begin{equation} 
	|P_{\delta,\varepsilon}(r)| \leq \max\{|r|,\delta\}
	\label{eq:Pinq}
\end{equation} 
for all $r\in \bbR$. 

\begin{lemma}
	If $M$ is a compact manifold and $f\in C^0(M;\bbR)$, then $\lim_{\delta\to 0^+} \sup_{\varepsilon \in (0,\delta/2)} \lVert |f|- P_{\delta,\varepsilon}\circ f \rVert_{C^0(M)}=0$. 
	\label{lem:conv_lemma}
\end{lemma}
\begin{proof}
	For each $\delta$, let $S_1 = \{p\in M: |f| < 3\delta/2 \}$ and $S_2 = \{p\in M: |f|\geq 3\delta/2\}$. Then, $M=S_1\sqcup S_2$. On $S_2$, $|f|-P_{\delta,\varepsilon}\circ f$ vanishes identically. On $S_1$,   $||f| - P_{\delta,\varepsilon}\circ f| < 3\delta$, via \cref{eq:Pinq}. 
\end{proof}

Thus, if $f\in C^K(M;\bbR)$, then $P_{\delta,\varepsilon}\circ f \in C^K(M)$ is a \emph{nonvanishing}, complex-valued  approximation to $|f|$ in the $C^0(M)$-norm.

Our next goal is to define a preliminary multicontour 
\begin{equation} 
Z^\circ:\mathsf{2}(A_{\ell,m,n})\to \bbC^N,
\label{eq:misc_058}
\end{equation} 
which, while not landing in $\calM_N$, lands in $(\bbC\backslash \{0,1\})^N$. So, the components of the to-be-defined map are not necessarily always distinct, but they do avoid $0,1$.
Defining such a contour is not difficult to do --- a product of Pochhammer contours would have the same property --- but the particular $Z^\circ$ defined below has the advantage of being perturbable to a multicontour $Z:\mathsf{2}(A_{\ell,m,n})\to \calM_N$. The `$\circ$' appearing as a superscript below signals preliminarity.

Just as in the previous section it was easier to work with the coordinates $a = (a_1,\ldots,a_N)$ than the coordinates $z\in \bbC^N$, where  $a_j$ and $z_j$ were related by a linear fractional transformation, it will be easier to define the desired map $Z^\circ$ by first defining a map
\begin{equation} 
	A^\circ : \mathsf{2}(A_{\ell,m,n})\to \bbC^N
	\label{eq:misc_059a}
\end{equation}
whose components are to be related to those of $Z^\circ$ by the same linear fractional transformations as $a$ to $z$.

The components of $A^\circ$ are defined by explicit formulas in large neighborhoods of the lifts $(\mathrm{bd}\circ \Pi )^{-1}(\mathrm{c})$ of the $2^N$ corners $\mathrm{c}\in \square^{\ell,m,n}$, which are labeled by triples $(\calS_1,\calS_2,\calS_3)$ of subsets $\calS_1\subseteq \calI_1$, $\calS_2\subseteq \calI_2$, and $\calS_3\subseteq \calI_3$. For each such triple, and for each $j\in \{1,\cdots,N\}$, $\epsilon\in (0,1/3)$, $\delta>0$, and $\varepsilon \in (0,\delta/2)$, define a function $A_j^\circ[\calS_1,\calS_2,\calS_3,\epsilon,\delta,\varepsilon] : \mathsf{2}(A_{\ell,m,n}) \to \bbC$ as follows: 
\begin{equation}
	A_j^\circ[\calS_1,\calS_2,\calS_3,\epsilon,\delta,\varepsilon] = 
	\begin{cases}
		\prod_{j\in \calS \subseteq \calI_1\cup \calI_3} P_{\delta,\varepsilon}(\tilde{\varrho}_{\mathrm{F}_{\calS;\infty},\epsilon}) & (j\in \calS_1), \\ 
		\prod_{j\in \calS \subseteq \calI_2\cup \calI_1} P_{\delta,\varepsilon}(\tilde{\varrho}_{\mathrm{F}_{\calS;0},\epsilon}) & (j\in \calS_2), \\   
		\prod_{j\in \calS \subseteq \calI_3\cup \calI_2} P_{\delta,\varepsilon}(\tilde{\varrho}_{\mathrm{F}_{\calS;1},\epsilon}) & (j\in \calS_3)
	\end{cases}
\end{equation}
if $j\in \calS_1\cup \calS_2\cup \calS_3$, while 
\begin{equation}
	1-A_j^\circ[\calS_1,\calS_2,\calS_3,\epsilon,\delta,\varepsilon] = 
	\begin{cases}
		\prod_{j\in \calS \subseteq \calI_2\cup \calI_1} P_{\delta,\varepsilon}(\tilde{\varrho}_{\mathrm{F}_{\calS;0},\epsilon}) & (j\in \calS_1^\complement), \\ 
		\prod_{j\in \calS \subseteq \calI_3\cup \calI_2}	 P_{\delta,\varepsilon}(\tilde{\varrho}_{\mathrm{F}_{\calS;1},\epsilon}) & (j\in \calS_2^\complement), \\   
		\prod_{j\in \calS \subseteq \calI_1\cup \calI_3} P_{\delta,\varepsilon}(\tilde{\varrho}_{\mathrm{F}_{\calS;\infty},\epsilon}) & (j\in \calS_3^\complement)
	\end{cases}
\end{equation}
if $j\notin  \calS_1\cup \calS_2\cup \calS_3$. Evidently, each $A_j^\circ[\calS_1,\calS_2,\calS_3,\epsilon,\delta,\varepsilon]$ is in $C^\infty(\mathsf{2}(A_{\ell,m,n}))$. 
The only dependence of $A_j^\circ[\calS_1,\calS_2,\calS_3,\epsilon,\delta,\varepsilon]$ on $\calS_1,\calS_2,\calS_3$ is through whether or not $j\in \calS_1\cup \calS_2\cup \calS_3$. 

\begin{figure}[t]
	\begin{center}
		\includegraphics[scale=.45]{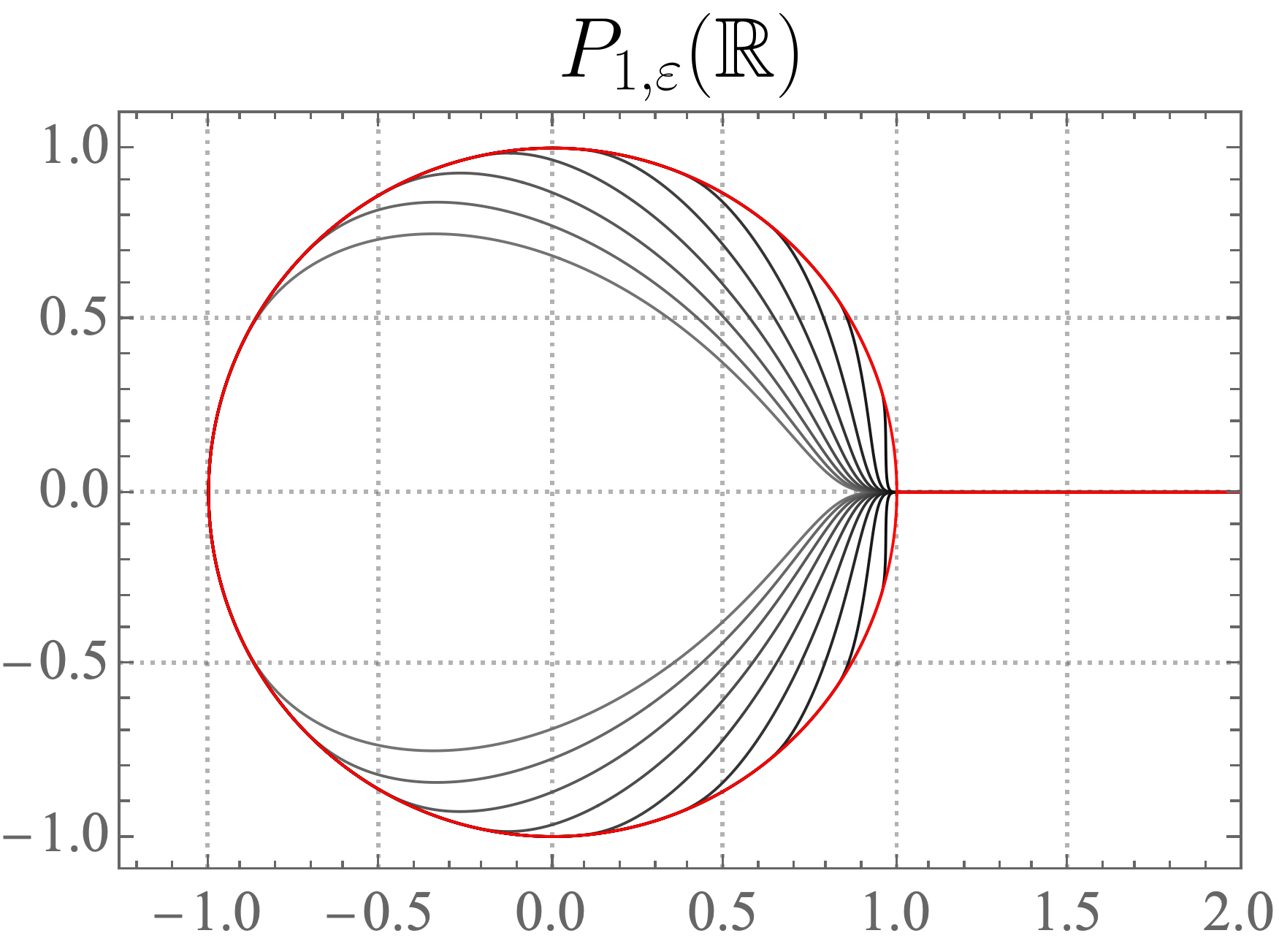}
	\end{center}
	\caption{The traces $P_{1,\varepsilon}(\bbR)\subseteq \bbC$ of $P_{1,\varepsilon}$ for $\varepsilon \in \{.05,\cdots,.45\}$, in various shades of gray, and of $P_1$, in red. Computed with $\Theta_{\mathrm{reg}}$ as in \cref{eq:misc_040}.}
\end{figure}

Let 
\begin{multline} 
	C[\calS_1,\calS_2,\calS_3] = \Pi^{-1}(\mathrm{cl}_{A_{\ell,m,n}} \{(a_1,\cdots,a_N)\in \square^{N\circ}_a: a_j<2/3\text{ if }j\in \calS_1\cup \calS_2\cup \calS_3  \\ \text{ and }a_j>1/3\text{ otherwise}\} ).
\end{multline} 
The interiors $C[\calS_1,\calS_2,\calS_3]^\circ$ cover $\mathsf{2}(A_{\ell,m,n})$. To see this, it suffices to note that $C[\calS_1,\calS_2,\calS_3]^\circ$ is the preimage of $\{(a_1,\cdots,a_N)\in \square^N_a: a_j<2/3\text{ if }j\in \calS_1\cup \calS_2\cup \calS_3 \text{ and }a_j>1/3\text{ otherwise}\}$ under $\mathrm{bd}\circ \Pi: \mathsf{2}(A_{\ell,m,n})\to \square_a^N$.  

The following proposition defines the map in \cref{eq:misc_059a} via stitching together the various locally defined maps above:
\begin{proposition}
	There exists some $\delta_0>0$ such that if $\epsilon \in (0,1/9)$, $\delta \in (0,\delta_0)$,  and $\varepsilon \in (0,\delta/2)$, then, given any $\calS_1,\calS_1'\subseteq \calI_1$, $\calS_2,\calS_2'\subseteq \calI_2$, and $\calS_3,\calS_3'\subseteq \calI_3$, 
	\begin{equation} 
		A_j^\circ[\calS_1,\calS_2,\calS_3,\epsilon,\delta,\varepsilon]=A_j^\circ[\calS_1',\calS_2',\calS_3',\epsilon,\delta,\varepsilon]
	\end{equation} 
	holds on $C[\calS_1,\calS_2,\calS_3] \cap C[\calS_1',\calS_2',\calS_3']$, for all $j\in \{1,\cdots,N\}$, so that there exists an element $A_j^\circ[\epsilon,\delta,\varepsilon]\in C^\infty(\mathsf{2}(A_{\ell,m,n});\bbC)$ whose restriction to the $C[\calS_1,\calS_2,\calS_3]$ are the $A_j^\circ[\calS_1,\calS_2,\calS_3,\epsilon,\delta,\varepsilon]$.
\end{proposition}
The idea is that $A_j^\circ[\calS_1,\calS_2,\calS_3,-]$ only differs from $a_j$ close to $\{a_j=0\}\cup\{a_j=1\}$, so in-between -- which is the only region in which the proposition is not tautological -- the function $A_j^\circ[\calS_1,\calS_2,\calS_3,-]$ does not depend on any of the $\calS$.
\begin{proof}
	The $A_j^\circ[\calS_1,\calS_2,\calS_3,\epsilon,\delta,\varepsilon],A_j^\circ[\calS_1',\calS_2',\calS_3',\epsilon,\delta, \varepsilon]$ automatically agree unless $j\in (\calS_1\cup \calS_2\cup\calS_3) \Delta (\calS_1'\cup \calS_2'\cup\calS_3')$. Suppose that $j\in \calS_1\cup \calS_2\cup\calS_3$ and $j\notin \calS_1'\cup \calS_2'\cup\calS_3'$. 
	
	On $C=C[\calS_1,\calS_2,\calS_3] \cap C[\calS_1',\calS_2',\calS_3']$, both of $a_j \circ \mathrm{bd}\circ \Pi$ and $1-a_j\circ \mathrm{bd}\circ \Pi$ are bounded from below by $1/3$.
	Consequently, because $\epsilon<1/9$, we have $\psi_\epsilon(a_j\circ \mathrm{bd}\circ \Pi) = a_j\circ \mathrm{bd}\circ \Pi$ and $\psi_\epsilon(1-a_j\circ \mathrm{bd}\circ \Pi) = 1-a_j\circ \mathrm{bd}\circ \Pi$ on $C$. 
	Moreover, for any $\calS\subseteq \{1,\cdots,N\}$ satisfying $\calS\ni j$, we have 
	\begin{equation} 
		\inf_{p\in C} \inf_{\epsilon \in (0,1/9)} |\tilde{\varrho}_{\mathrm{F}_{\calS;x_0},\epsilon}(p)| >0.
	\end{equation}  
	It follows that, if $\delta$ is sufficiently small, 
	\begin{align}
	\begin{split}
	A_j^\circ[\calS_1,\calS_2,\calS_3,\epsilon,\delta,\varepsilon] &= 
	\begin{cases}
		\prod_{j\in \calS \subseteq \calI_1\cup \calI_3}\varrho_{\mathrm{F}_{\calS;\infty},\epsilon}\!\circ \Pi & (j\in \calS_1), \\ 
		\prod_{j\in \calS \subseteq \calI_2\cup \calI_1} \varrho_{\mathrm{F}_{\calS;0},\epsilon}\,\circ \Pi & (j\in \calS_2), \\   
		\prod_{j\in \calS \subseteq \calI_3\cup \calI_2} \varrho_{\mathrm{F}_{\calS;1},\epsilon}\,\circ \Pi & (j\in \calS_3),
	\end{cases}  \\ 
	1- A_j^\circ[\calS_1',\calS_2',\calS_3',\epsilon,\delta,\varepsilon] &= 
	 \begin{cases}
	 	\prod_{j\in \calS \subseteq \calI_2\cup \calI_1} \varrho_{\mathrm{F}_{\calS;0},\epsilon} \,\circ \Pi & (j\in \calS_1^\complement), \\ 
	 	\prod_{j\in \calS \subseteq \calI_3\cup \calI_2} \varrho_{\mathrm{F}_{\calS;1},\epsilon}\,\circ \Pi & (j\in \calS_2^\complement), \\   
	 	\prod_{j\in \calS \subseteq \calI_1\cup \calI_3} \varrho_{\mathrm{F}_{\calS;\infty},\epsilon}\!\circ \Pi & (j\in \calS_3^\complement).
	 \end{cases}
 	\end{split} 
	\end{align} 
	on $C$, since $|\tilde{\varrho}_{\mathrm{F}_{\calS;x_0},\epsilon}|=\varrho_{\mathrm{F}_{\calS;x_0},\epsilon}\circ \Pi$. 
	Consequently, by \cref{eq:product_identity}, we have 
	\begin{equation} 
		A_j^\circ[\calS_1,\calS_2,\calS_3,\epsilon,\delta,\varepsilon] = \psi_\epsilon\circ a_j \circ \mathrm{bd}\circ \Pi = a_j \circ \mathrm{bd}\circ \Pi
	\end{equation} 
	and 
	\begin{equation}
		1-A_j^\circ[\calS_1',\calS_2',\calS_3',\epsilon,\delta,\varepsilon] = \psi_\epsilon(1-a_j\circ\mathrm{bd}\circ\Pi)=1-a_j\circ\mathrm{bd}\circ\Pi
	\end{equation} 
	on $C$. 
	So, the functions $A_j^\circ[\calS_1,\calS_2,\calS_3,\epsilon,\delta,\varepsilon]$ and $A_j^\circ[\calS_1',\calS_2',\calS_3',\epsilon,\delta,\varepsilon]$ agree there. 
\end{proof}

If $j\in \calI_1$, let $Z_j^\circ[\epsilon,\delta,\varepsilon ] = -(1- A_j^\circ[\epsilon,\delta,\varepsilon ]) / A_j^\circ[\epsilon,\delta,\varepsilon ]$, if $j\in \calI_2$, let $Z_j^\circ[\epsilon,\delta,\varepsilon ] =  A_j^\circ[\epsilon,\delta,\varepsilon ]$, and if $j\in \calI_3$, let $Z_j^\circ[\epsilon,\delta,\varepsilon ] =  1/(1-A_j^\circ[\epsilon,\delta,\varepsilon ])$. 

We now check that $A^\circ,Z^\circ$ are, in fact, avoiding $0,1$. 

\begin{proposition}
	There exists some $\delta_1 \in (0,\delta_0]$ such that, for each $\epsilon \in (0,1/10)$, $\delta\in (0, \delta_1)$, and $\varepsilon \in (0,\delta/2)$,  $\smash{A_j^\circ[\epsilon,\delta,\varepsilon](p)}\notin \{0,1\}$
	for all $p\in \mathsf{2}(A_{\ell,m,n})$, and therefore the same applies to the $Z_j^\circ[\epsilon,\delta,\varepsilon]$. 
\end{proposition}

\begin{proof} Since $P_{\delta,\varepsilon}$ is nonvanishing, $A_j^\circ[\epsilon,\delta,\varepsilon]$ is novanishing on $C[\calS_1,\calS_2,\calS_3]$ if $j\in \calS=\calS_1\cup \calS_2\cup \calS_3$. Similarly, $1-A_j^\circ[\epsilon,\delta,\varepsilon]$ is nonvanishing on  $C[\calS_1,\calS_2,\calS_3]$ if $j\notin \calS$. In short, taking $\delta \to 0^+$, we have
\begin{equation} 
		A_j^\circ [\epsilon,\delta,\varepsilon] \to  a_j\circ \mathrm{bd}\circ \Pi
\end{equation} 
in $C^0(\mathsf{2}(A_{\ell,m,n}))$, uniformly in $\epsilon ,\varepsilon$ (at least if $\epsilon$ stays away from $1/9$), and since $a_j\circ \mathrm{bd}\circ \Pi \leq 2/3$ on $C[\calS_1,\calS_2,\calS_3]$ if $j\in \calS_1\cup \calS_2\cup \calS_3$, this means that taking $\delta$ sufficiently small we force $A_j^\circ[\epsilon,\delta,\varepsilon] \neq 1$ on $C[\calS_1,\calS_2,\calS_3]$. 
	
Here are the details. Let 
\begin{equation} 
		\delta_1 =  \inf_{\epsilon \in (0,1/10)} \min\Big\{\delta_0, \frac{1}{2}  \prod_{\mathrm{F}\in \calF(A_{\ell,m,n})}(1+ \lVert \varrho_{\mathrm{F},\epsilon} \rVert_{L^\infty(A_{\ell,m,n})})^{-1} \Big\}.
\end{equation} 
This is nonzero because the family $\{\varrho_{\mathrm{F},\epsilon}\}_{\epsilon \in (0,1/9)}\subset L^\infty(A_{\ell,m,n})$ extends continuously down to $\epsilon = 0$.
	 
For $j\in \calS$ and $\delta\in (0,\delta_1)$, \cref{eq:Pinq}, together with \cref{eq:misc_975} and \cref{eq:product_identity}, implies that, on $C[\calS_1,\calS_2,\calS_3]$, 
\begin{equation} 
		|A_j^\circ[\epsilon,\delta,\varepsilon]| \leq \max\Big\{\frac{1}{2}, \psi_\epsilon (a_j \circ \mathrm{bd}\circ \Pi  )\Big\} \leq \max\Big\{\frac{1}{2}, a_j \circ \mathrm{bd}\circ \Pi  \Big\}  \leq \frac{2}{3},
\end{equation} 
so $A_j^\circ[\epsilon,\delta,\varepsilon]$ cannot equal $1$. Likewise, if $j\notin \calS$, then $|1-A_j^\circ[\epsilon,\delta,\varepsilon]| < 1$, so $A_j^\circ[\epsilon,\delta,\varepsilon]$ cannot equal $0$. 
\end{proof} 

We now begin the process of modifying $A^\circ,Z^\circ$ so as to define a multicontour landing in $\calM_N$. As stated above, the obstruction is that the components of $Z^\circ$ are not always distinct. 

The only issue is pairs $Z_j^\circ,Z_k^\circ$ for $j,k$ lying in the same member of $\{\calI_1,\calI_2,\calI_3\}$:
\begin{proposition} 
	There exists some $\delta_2 \in (0,\delta_1]$ such that, for each $\epsilon \in (0,1/10)$, $\delta\in (0,\delta_2)$, and $\varepsilon \in (0,\delta/2)$, if $j,k$ are distinct elements of $\{1,\cdots,N\}$ that are not in the same member of $\{\calI_1,\calI_2,\calI_3\}$, then 
\begin{equation} 
	Z_j^\circ[\epsilon,\delta,\varepsilon](p)\neq Z_k^\circ[\epsilon,\delta,\varepsilon](p)
	\label{eq:misc_zjzk}
\end{equation} 
for every $p\in \mathsf{2}(A_{\ell,m,n})$. 
\label{prop:differance}
\end{proposition} 
The basic idea, which we illustrate when $j\in \calI_2$ and $k\in \calI_3$, is that the difference $Z_j^\circ - Z_k^\circ$ is, near the set $(\mathrm{bd}\circ \Pi)^{-1}(\{z_j,z_k = 0\}) $ where \cref{eq:misc_zjzk} could conceivably fail,  the product of the non-vanishing quantity 
\begin{equation}
\prod_{j,k\in \calS\subseteq \calJ\cup \calK} P_{\delta,\varepsilon}(\tilde{\varrho}_{\mathrm{F}_{\calS;0},\epsilon})
\end{equation}
and a function whose real part is positive, with similar statements holding in the other cases. The key point is proving positivity. 
\begin{proof} 
Let $j\in \calJ$ and $k\in \calK$ denote distinct elements of $\{1,\cdots,N\}$ in different members $\calJ,\calK\in\{\calI_1,\calI_2,\calI_3\}$, and assume without loss of generality that $j<k$.  We now define $y_{j,k;\epsilon} \in C^\infty(A_{\ell,m,n};\bbR^+)$ measuring the separation between the $j$th and $k$th coordinates. Three cases need to be considered:
\begin{itemize}
	\item If $\calJ=\calI_1$ and $\calK=\calI_2$, let 
	\begin{equation} 
		y_{j,k;\epsilon} = \Big(\prod_{j,k\in \calS\subseteq \calJ\cup \calK} \varrho_{\mathrm{F}_{\calS;0},\epsilon} \Big)^{-1}  (1-x_j)^{-1} (x_k-x_j),
	\end{equation} 
	which, in terms of $a_j,a_k$ is given by $(\prod_{j,k\in \calS\subseteq \calJ\cup \calK} \varrho_{\mathrm{F}_{\calS;0},\epsilon} )^{-1}  (-a_k (1-a_j) +a_k + 1 - a_j )$, 
	\item if $\calJ = \calI_2$ and $\calK=\calI_3$, let 
	\begin{equation} 
		y_{j,k;\epsilon} = \Big(\prod_{j,k\in \calS\subseteq \calJ\cup \calK} \varrho_{\mathrm{F}_{\calS;1},\epsilon} \Big)^{-1}  x_k^{-1} (x_k-x_j),
	\end{equation} 
	which, in terms of $a_j,a_k$ is given by the same formula as in the previous case, 
	\item if $\calJ = \calI_1$ and $\calK = \calI_3$, let 
	\begin{equation} 
		y_{j,k;\epsilon} = (\prod_{j,k\in \calS\subseteq \calJ\cup \calK} \varrho_{\mathrm{F}_{\calS;\infty},\epsilon} )^{-1}  x_k^{-1}(1- x_j)^{-1} (x_k-x_j),
	\end{equation} 
	which, in terms of $a_j,a_k$, is given by the same formula as in the previous two cases except with $j,k$ switched: 
	\begin{equation} 
		y_{j,k;\epsilon} = (\prod_{j,k\in \calS\subseteq \calJ\cup \calK} \varrho_{\mathrm{F}_{\calS;\infty},\epsilon} )^{-1} (- a_j (1-a_k)+a_j + 1-a_k ).
	\end{equation} 
\end{itemize} 
That the $y_{j,k;\epsilon}$ lie in $C^\infty(A_{\ell,m,n})$ can be checked in local coordinates. More important is the observation (which just follows from factoring out copies of bdfs from $a_j,1-a_k$ or $1-a_j,a_k$) that $y_{j,k;\epsilon} \in C^0(A_{\ell,m,n};\bbR^+)$, as this implies that $\inf y_{j,k;\epsilon}>0$ on $A_{\ell,m,n}$.

Let $e_k = -1$ if $k\in \calI_3$ and $e_k = 0$ otherwise, and let $e_j=-1$ if $j\in \calI_1$ and $e_j=0$ otherwise.
Then, as argued below,
\begin{multline} 
	\lim_{\delta \to 0^+} \sup_{\epsilon \in (0,1/10)}\sup_{\varepsilon \in (0,\delta/2)} \Big\lVert  \Big(\prod_{j,k\in \calS\subseteq \calJ\cup \calK} P_{\delta,\varepsilon}(\tilde{\varrho}_{\mathrm{F}_{\calS;x_0},\epsilon}) \Big)^{-1}Z_k^{\circ e_k} (1-Z_j^{\circ})^{e_j}(Z_k^\circ - Z_j^\circ) \\ - y_{j,k;\epsilon}\circ \Pi  \Big\rVert_{L^\infty(\mathsf{2}(A_{\ell,m,n}))} = 0. 
	\label{eq:misc_057}
\end{multline} 
Thus, as long as $\delta$ is sufficiently small, 
\begin{equation}
 	\Big(\prod_{j,k\in \calS\subseteq \calJ\cup \calK} P_{\delta,\varepsilon}(\tilde{\varrho}_{\mathrm{F}_{\calS;x_0},\epsilon}) \Big)^{-1}Z_k^{\circ e_k} (1-Z_j^{\circ})^{e_j}(Z_k^\circ - Z_j^\circ)  >0 , 
 	\label{eq:misc_075}
\end{equation}
and this necessitates that $Z_j^\circ[\epsilon,\delta,\varepsilon](p)\neq Z_k^\circ[\epsilon,\delta,\varepsilon](p)$ for all $p\in \mathsf{2}(A_{\ell,m,n})$.

In order to prove \cref{eq:misc_057}, first note that, given any neighborhood 
\begin{equation} 
	U\supseteq \Pi^{-1}(\cup_{j,k\in \calS\subseteq \calJ\cup \calK} \mathrm{F}_{\calS;x_0}),
\end{equation} 
it is certainly the case, via \Cref{lem:conv_lemma}, that \cref{eq:misc_057} holds when the $L^\infty(\mathsf{2}(A_{\ell,m,n}))$-norm is replaced by the $L^\infty(\mathsf{2}(A_{\ell,m,n})\backslash U)$-norm. 
So, it suffices to check the situation near the excised set $\Pi^{-1}(\cup_{j,k\in \calS\subseteq \calJ\cup \calK} \mathrm{F}_{\calS;x_0})$. 

Near this set, we use the definitions of $Z_j^\circ[\epsilon,\delta,\varepsilon]$ and $Z_k^\circ[\epsilon,\delta,\varepsilon]$. All of the $P_{\delta,\varepsilon}(\tilde{\varrho}_{\mathrm{F},\epsilon})$ factors in the denominator on the left-hand side of \cref{eq:misc_075} cancel out with corresponding factors in the definitions of $Z_j^\circ,Z_k^\circ$. For instance, when $\calJ=\calI_1$ and $\calK=\calI_2$, 
\begin{multline}
	  \Big(\prod_{j,k\in \calS\subseteq \calJ\cup \calK} P_{\delta,\varepsilon}(\tilde{\varrho}_{\mathrm{F}_{\calS;x_0},\epsilon}) \Big)^{-1}Z_k^{\circ e_k} (1-Z_j^{\circ})^{e_j}(Z_k^\circ - Z_j^\circ)  \\ =  \Big(\prod_{j,k\in \calS\subseteq \calJ\cup \calK} P_{\delta,\varepsilon}(\tilde{\varrho}_{\mathrm{F}_{\calS;x_0},\epsilon}) \Big)^{-1} (-A_k^\circ(1-A_j^\circ) + A_k^\circ + 1-A_j^\circ ),
\end{multline}
and the factors of $\smash{P_{\delta,\varepsilon}(\tilde{\varrho}_{\mathrm{F}_{\calS;x_0},\epsilon})^{-1}}$ cancel with terms in the definition of $1-A_j^\circ,A_k^\circ$.

A similar statement applies to $y_{j,k;\epsilon}$, with the result being the same end result (which depends on what $\calJ,\calK$ are) without the ``$P_{\delta,\varepsilon}$'' surrounding the $\tilde{\varrho}_{\mathrm{F},\epsilon}$. For instance, in the $\calJ=\calI_1,\calK=\calI_2$ case,
\begin{equation}
	y_{j,k;\epsilon} = \Big(\prod_{j,k\in \calS\subseteq \calJ\cup \calK} \tilde{\varrho}_{\mathrm{F}_{\calS;x_0},\epsilon}\Big)^{-1} (-A_k^\circ(1-A_j^\circ) + A_k^\circ + 1-A_j^\circ ).
\end{equation}
The claim then follows from \Cref{lem:conv_lemma}.
\end{proof} 

A similar computation yields: 
\begin{lemmap}
	Let $j<k$ denote distinct elements of the \emph{same} member of $\{\calI_1,\calI_2,\calI_3\}$. 
	Then, letting $y_{j,k;\epsilon} \in C^\infty(A_{\ell,m,n};\bbR)$ be defined by 
	\begin{equation}
		a_k - a_j = y_{j,k;\epsilon} \prod_{x_0\in \{\infty,0,1\}}\prod_{j,k\in \calS} \varrho_{\mathrm{F}_{\calS;x_0},\epsilon}, 
	\end{equation}
	where the second product is over the subsets  $\calS\subseteq \{1,\cdots,N\}$ such that the $\varrho_{\mathrm{F}_{\calS;x_0},\epsilon}$ are defined, we have 
	\begin{multline} 
		\lim_{\delta \to 0^+}  \sup_{\epsilon \in (0,1/10)} \sup_{\varepsilon \in (0,\delta/2)} \Big\lVert  \Big(\prod_{x_0\in \{\infty,0,1\}}\prod_{j,k\in \calS} P_{\delta,\varepsilon}(\tilde{\varrho}_{\mathrm{F}_{\calS;x_0},\epsilon}) \Big)^{-1}\Big(A_k^\circ[\epsilon,\delta,\varepsilon] - A_j^\circ[\epsilon,\delta,\varepsilon]\Big) \\ - y_{j,k;\epsilon }\circ \Pi \Big\rVert_{L^\infty(\mathsf{2}(A_{\ell,m,n}))} = 0. 
	\end{multline}  
	\label{lem:misc}
\end{lemmap}
Note that unlike the $y_{j,k;\epsilon}$ appearing in the proof of the previous proposition, the $y_{j,k;\epsilon}$ appearing in \Cref{lem:misc} attain both signs on $A_{\ell,m,n}$.

Finally, $A$ can be defined. Let 
\begin{equation} 
	A_j[\epsilon,\delta,\varepsilon,\digamma] = A_j^\circ[\epsilon,\delta,\varepsilon] \times e^{i j/\digamma}
\end{equation}  
for $\digamma>0$.
If $j\in \calI_1$, let $Z_j[\epsilon,\delta,\varepsilon,\digamma] = -(1- A_j[\epsilon,\delta,\varepsilon,\digamma ]) / A_j[\epsilon,\delta,\varepsilon,\digamma ]$, if $j\in \calI_2$, let $Z_j[\epsilon,\delta,\varepsilon,\digamma ] =  A_j[\epsilon,\delta,\varepsilon,\digamma ]$, and for $j\in \calI_3$, let $j\in \calI_2$, let $Z_j[\epsilon,\delta,\varepsilon,\digamma] =  1/(1-A_j[\epsilon,\delta,\varepsilon,\digamma ])$.

This works:
\begin{proposition}
	For all $\epsilon \in (0,1/10)$, there exists some $\delta_3(\epsilon)\in (0,\delta_2)$ such that, for all $\delta \in (0,\delta_3)$ and $\varepsilon \in (0,\delta/2)$, there exists some $\digamma_0(\epsilon,\delta,\varepsilon)>0$ such that, for all $\digamma>\digamma_0$, the map $Z[\epsilon,\delta,\varepsilon,\digamma] \in C^\infty(\mathsf{2}(A_{\ell,m,n}); \bbC^N)$ defined by \begin{equation} 
		Z[\epsilon,\delta,\varepsilon,\digamma]=(Z_1[\epsilon,\delta,\varepsilon,\digamma],\cdots,Z_N[\epsilon, \delta,\varepsilon,\digamma])
	\end{equation}
	has image lying in $\calM_N(0,\infty)$. 
	\label{prop}
\end{proposition}
\begin{proof}
	We need to check that, for $\delta,\varepsilon$ sufficiently small and $\digamma>0$ sufficiently large, three things hold for all $p\in \mathsf{2}(A_{\ell,m,n})$: (1)  $Z_j[\epsilon,\delta,\varepsilon,\digamma](p) \neq 0$, (2) $Z_j[\epsilon,\delta,\varepsilon,\digamma](p) \neq 1$, and (3) $Z_j[\epsilon,\delta,\varepsilon,\digamma](p) \neq Z_k[\epsilon,\delta,\varepsilon,\digamma](p)$ for any distinct $j,k\in \{1,\cdots,N\}$. 
	The first of these follows immediately from the fact that $Z_j^\circ [\epsilon,\delta,\varepsilon]$, and therefore $Z_j[\epsilon,\delta,\varepsilon]$, is non-vanishing. 
	The second follows from the observation that
	\begin{equation} 
		\lim_{\digamma\to \infty} \lVert Z_j[\epsilon,\delta,\varepsilon,\digamma] - Z_j^\circ[\epsilon,\delta,\varepsilon] \rVert_{L^\infty(\mathsf{2}(A_{\ell,m,n}))} = 0,
		\label{eq:misc_059}
	\end{equation} 
	which implies the claim because $\smash{1-Z_j^\circ[\epsilon,\delta,\varepsilon]}$ is nonvanishing on $\mathsf{2}(A_{\ell,m,n})$ (and is therefore bounded away from $0$, since $\mathsf{2}(A_{\ell,m,n})$ is compact). 
	So, it remains to check that, for all $p\in \mathsf{2}(A_{\ell,m,n})$, $Z_j[\epsilon,\delta,\varepsilon,\digamma](p) \neq Z_k[\epsilon,\delta,\varepsilon,\digamma](p)$ for any distinct $j,k\in \{1,\cdots,N\}$. 
	
	If $j,k$ are in different members of $\{\calI_1,\calI_2,\calI_3\}$, then for all $\delta \in (0,\delta_2)$ and $\varepsilon \in (0,\delta/2)$, \Cref{prop:differance}, together with \cref{eq:misc_059}, implies that as long as $\digamma$ is sufficiently large, $Z_j[\epsilon,\delta,\varepsilon,\digamma](p) \neq Z_k[\epsilon,\delta,\varepsilon,\digamma](p)$ for any $p\in \mathsf{2}(A_{\ell,m,n})$.
	
	The remaining case, when $j,k$ are in the same member of $\{\calI_1,\calI_2,\calI_3\}$, is more delicate. Let $O_{j,k}\subset \mathsf{2}(A_{\ell,m,n})$ denote an open neighborhood of $\Pi^{-1}(H_{j,k})$ 
	\begin{itemize}
		\item whose closure intersects only those lifted faces $\Pi^{-1}(\mathrm{F})$  for which $H_{j,k}$ intersects $\mathrm{F}$, i.e.\ those of the form $\mathrm{F}_{\calS,x_0}$ for $\calS\subseteq \{1,\cdots,N\}$ such that $j,k\in \calS$, and 
		\item that is covered by the sets $C[\calS_1,\calS_2,\calS_3]$ for $j,k$ either both in $\calS_1\cup \calS_2\cup \calS_3$ or both not (whose interiors form an open cover of $H_{j,k}$). 
	\end{itemize}
	These properties can be arranged simultaneously. 
	
	Because $O_{j,k}$ contains $\Pi^{-1}(H_{j,k})$, which is the vanishing set of $y_{j,k;\epsilon}$, \Cref{lem:misc} implies that, as long as $\delta$ is sufficiently small and $\digamma$ is sufficiently large (depending on $\delta$), the zeroes of  $A_k[\epsilon,\delta,\varepsilon,\digamma]- A_j[\epsilon,\delta,\varepsilon,\digamma]$ will be contained entirely in this set. Since the linear fractional transformations defining $Z_j$ from $A_j$ are injective, the same therefore holds for $Z_k[\epsilon,\delta,\varepsilon,\digamma]- Z_j[\epsilon,\delta,\varepsilon,\digamma]$.

	In the set $C[\calS_1,\calS_2,\calS_3]$, if $j,k\in \calS_1\cup \calS_2\cup \calS_3$,
	\begin{multline}
		A_k[\epsilon,\delta,\varepsilon,\digamma] -A_j[\epsilon,\delta,\varepsilon,\digamma] \\ = 
		\begin{cases}
			e^{ik/\digamma}\prod_{k\in \calS \subseteq \calI_1\cup \calI_3} P_{\delta,\varepsilon}(\tilde{\varrho}_{\mathrm{F}_{\calS;\infty},\epsilon}) - 	e^{ij/\digamma}\prod_{j\in \calS \subseteq \calI_1\cup \calI_3} P_{\delta,\varepsilon}(\tilde{\varrho}_{\mathrm{F}_{\calS;\infty},\epsilon}) & (j\in \calS_1), \\ 
			e^{ik/\digamma}\prod_{k\in \calS \subseteq \calI_2\cup \calI_1} P_{\delta,\varepsilon}(\tilde{\varrho}_{\mathrm{F}_{\calS;0},\epsilon})  \;- e^{ij/\digamma}\prod_{j\in \calS \subseteq \calI_2\cup \calI_1} P_{\delta,\varepsilon}(\tilde{\varrho}_{\mathrm{F}_{\calS;0},\epsilon})& (j\in \calS_2), \\   
			e^{ik/\digamma}\prod_{k\in \calS \subseteq \calI_3\cup \calI_2} P_{\delta,\varepsilon}(\tilde{\varrho}_{\mathrm{F}_{\calS;1},\epsilon}) \;- e^{ij/\digamma}\prod_{j\in \calS \subseteq \calI_3\cup \calI_2} P_{\delta,\varepsilon}(\tilde{\varrho}_{\mathrm{F}_{\calS;1},\epsilon})  & (j\in \calS_3).
		\end{cases} 
	\end{multline}
	Thus, if this vanishes at some point $p\in O_{j,k} \cap C[\calS_1,\calS_2,\calS_3]$,  we have 
	\begin{equation}
		e^{i (k-j)/\digamma}  \prod_{\calS\text{ s.t.\,} k\in \calS, j\notin \calS} P_{\delta,\varepsilon}(\tilde{\varrho}_{\mathrm{F}_{\calS;x_0},\epsilon}(p)) =  \prod_{\calS\text{ s.t.\,} j\in \calS, k\notin \calS} P_{\delta,\varepsilon}(\tilde{\varrho}_{\mathrm{F}_{\calS;x_0},\epsilon}(p)), 
		\label{eq:misc_061}
	\end{equation}
	where $x_0 = \infty$ if $j,k\in \calI_1$, $x_0=0$ if $j,k\in \calI_2$, and $x_0 =1$ if $j,k\in \calI_3$. But, if $\delta$ is sufficiently large, then, owing to the assumption that the closure of $O_{j,k}$ intersects only those $\Pi^{-1}(\mathrm{F})$ for which $H_{j,k}$ intersects $\mathrm{F}$, each of the $P_{\delta,\varepsilon}(\tilde{\varrho}_{\mathrm{F}_{\calS;x_0},\epsilon})$ terms in \cref{eq:misc_061} satisfy  
	\begin{equation} 
		P_{\delta,\varepsilon}(\tilde{\varrho}_{\mathrm{F}_{\calS;x_0},\epsilon}) = \varrho_{\mathrm{F}_{\calS;x_0},\epsilon}\circ \Pi
	\end{equation} 
	on $O_{j,k}$. In particular, they are real-valued. Thus, since the right-hand side of \cref{eq:misc_061} is a nonzero real number and, if $\digamma$ is sufficiently large, the left-hand side is a nonreal complex number, we have reached a contradiction. The conclusion is that $A_k[\epsilon,\delta,\varepsilon,\digamma] -A_j[\epsilon,\delta,\varepsilon,\digamma]$ cannot vanish in $O_{j,k} \cap C[\calS_1,\calS_2,\calS_3]$.
	
	The situation in the set $C[\calS_1,\calS_2,\calS_3]$ if $j,k\notin \calS_1\cup \calS_2\cup \calS_3$ is analogous.
\end{proof}

\begin{remark}
	$\delta_3(\epsilon)$ can actually be taken independent of $\epsilon$, as the same argument with more bookkeeping shows.
\end{remark}

\section{Lifting lemma}

We now check that, for any $\epsilon \in (0,1/10)$, $\delta \in (0,\delta_3)$, $\varepsilon \in (0,\delta/2)$,  and $\digamma>\digamma_0$, the map $Z[\epsilon,\delta,\varepsilon,\digamma]: \mathsf{2}(A_{\ell,m,n})\to \calM_N(0,\infty)$, defined in the previous section above \Cref{prop}, lifts to a continuous map 
\begin{equation}
	\widehat{Z}[\epsilon,\delta,\varepsilon,\digamma] : \mathsf{2}(A_{\ell,m,n})\to \widehat{\calM}_N(0,\infty), 
\end{equation}
where recall that $\widehat{\calM}_N(0,\infty) = \widetilde{\calM}_N(0,\infty) / [\pi_1(\calM_N(0,\infty)),\pi_1(\calM_N(0,\infty))]$ is the monodromy cover of $\calM_N$ discussed in the introduction.

If $Z[\epsilon,\delta,\varepsilon,\digamma]$ lifts, this lift is automatically smooth, given that $Z[\epsilon,\delta,\varepsilon,\digamma]$ is. After all, the covering map 
\begin{equation} 
\widehat{\calM}_N(0,\infty)\to \calM_N(0,\infty)
\end{equation} 
is locally a diffeomorphism. 

The given map lifts if and only if, for arbitrary $p_0\in \mathsf{2}(A_{\ell,m,n})$, the induced map 
\begin{equation}
	Z[\epsilon,\delta,\varepsilon,\digamma]_* : \pi_1(\mathsf{2}(A_{\ell,m,n}),p_0) \to \pi_1(\calM_N(0,\infty),p_1),
	\label{eq:misc_062}
\end{equation}
where $p_1 = Z[\epsilon,\delta,\varepsilon,\digamma](p_0)\in \calM_N(0,\infty)$,
has image lying in the commutator subgroup of the codomain \cite[Prop.\ 1.33]{Hatcher}. 

Any two of these maps, for different values of $\delta,\varepsilon,\digamma$, are homotopic, via a homotopy dialing these parameters, so it suffices to check a single one. Taking $\epsilon \in (0,1/10)$ and $\delta$ small enough, we can assume that there exists some $p \in \smash{(\square^{N}_a)^\circ} \cap \calM_N(0,\infty)$ such that $A_j[\epsilon,\delta,\varepsilon,\digamma](p_0) = e^{ij/\digamma} a_j $ for all $p_0\in \Pi^{-1}(p)$, where $a_j$ is the $j$th coordinate of $p$ in $\square^N_a$. 
It follows that the map 
\begin{equation}
	Z[\epsilon,\delta,\varepsilon,\digamma]_* : \pi_1(\mathsf{2}(A_{\ell,m,n})) \to \pi_1(\calM_N(0,\infty))
\end{equation}
of groupoids restricts to a map 
\begin{equation} 
	\pi_1(\mathsf{2}(A_{\ell,m,n}),\Pi^{-1}(p)) \to \pi_1(\calM_N(0,\infty),p(\digamma))
\end{equation} 
of groupoids, where $p(\digamma)$ is the point in $\bbC^N_a$ with $j$th coordinate $e^{ij/\digamma}a_j$. 

For each $\mathrm{F}\in \calF(A_{\ell,m,n})$, it is the case that, as long as $\delta$ is sufficiently small, the image of $\gamma_{\calF,\mathrm{F}}:[0,1] \to \mathsf{2}(A_{\ell,m,n})$ under $Z[\epsilon,\delta,\varepsilon,\digamma]$ depends on the subset $\calF\subseteq \calF(A_{\ell,m,n})$ only through on whether or not $\mathrm{F}\in \calF$. Indeed, $\gamma_{\calF,\mathrm{F}}$ stays away from every $\Pi^{-1}(\mathrm{F}_0)$ for $\mathrm{F}_0\in \calF(A_{\ell,m,n})\backslash \{\mathrm{F}\}$, so, if $\delta$ is sufficiently small, then 
\begin{equation}
	P_{\delta,\varepsilon}(\tilde{\varrho}_{\mathrm{F}_0}) = \varrho_{\mathrm{F}_0}\circ \Pi 
\end{equation}
on its image. 

Thus, the $\calF$-dependence of $Z[\epsilon,\delta,\varepsilon,\digamma]$ on $\gamma_{\calF,\mathrm{F}}$ enters only through the $P_{\delta,\varepsilon}(\tilde{\varrho}_{\mathrm{F}})$ terms. But, this function, on the image of $\gamma_{\calF,\mathrm{F}}$, depends on $\calF$ only through on whether or not $\mathrm{F}\in \calF$. 
We can therefore conclude from \Cref{lem:lifting_lemma} that the image of \cref{eq:misc_062} lies in the desired commutator subgroup.

\section{Computation of pairing}

Finally, we check that, for an appropriate choice of branch of the lift $\widehat{Z}[\epsilon,\delta,\varepsilon,\digamma]$, 
\begin{multline}
	\int_{\mathsf{2}(A_{\ell,m,n})}  \widehat{Z}[\epsilon,\delta,\varepsilon,\digamma]^* \omega(\bmalpha,\bmbeta,\bmgamma) = \Big[ \prod_{\varnothing\subsetneq \calS\subseteq \{1,\cdots,\ell+m\}} 2i\sin(\pi \alpha_{\calS} )  \Big] \Big[ \prod_{\varnothing \subsetneq \calS \subseteq \{\ell+1,\cdots,N\}} 2i\sin(\pi \beta_{\calS})  \Big] \\ \times  \Big[ \prod_{\varnothing\subsetneq \calS\subseteq \{1,\cdots,\ell\}\cup \{\ell+m+1,\cdots,N\} } 2i\sin(\pi \zeta_{\calS})  \Big] I_{\ell,m,n}(\bmalpha,\bmbeta,\bmgamma),
\end{multline} 
which is a restatement of \cref{eq:fundamental} in the $r=0,R=\infty$ case.

Since $\omega$ is a closed form, as changing $\delta,\varepsilon,\digamma$ results in a homotopy of the mapping $\hat{Z}[\epsilon,\delta,\varepsilon,\digamma]$, it suffices to prove this result for a single value of these parameters, for each $\epsilon$.

Moreover, via the analytic dependence of $\omega(\bmalpha,\bmbeta,\bmgamma)$ on the parameters $\bmalpha,\bmbeta,\bmgamma$, it suffices to prove that the formula above holds when the components of $\bmgamma$ are all positive and $\alpha_j,\beta_j>0$ if $j\in \{\ell+1,\cdots,\ell+m\}$, $\alpha_j,\zeta_j>0$ if $j\in \{1,\cdots,\ell\}$, and $\beta_j,\zeta_j>0$ if $j\in \{\ell+m+1,\cdots,N\}$. Here, $\zeta_j = \zeta_{\{j\}}$, 
where $\zeta_{\calS}$ was defined in \cref{eq:zetadef}. These conditions pick out a nonempty, open conic subset's worth of 
\begin{equation} 
(\bmalpha,\bmbeta,\bmgamma) \in \bbC^{N}\times \bbC^N \times \bbC^{N(N-1)/2}.
\end{equation} 
Indeed, the conditions pick out an open conic subset, and it is nonempty because the ones involving $\bmalpha,\bmbeta$ are satisfied if 
\begin{itemize}
	\item $\alpha_j>0$ if $j\in \{1,\cdots,\ell+m\}$ and $-\alpha_j \gg \lVert \bmgamma \rVert + \beta_j$ if $j\in \{\ell+m+1,\cdots,N\}$,
	\item $\beta_j>0$ if $j\in \{\ell+1,\cdots,N\}$ and $-\beta_j\;\! \gg \lVert \bmgamma \rVert + \alpha_j$ if $j\in \{1,\cdots,\ell\}$,
\end{itemize}
where ``$\gg$'' means ``$\geq$'' with an unspecified large constant on the right-hand side.

For such $\bmalpha,\bmbeta,\bmgamma$, 
\begin{equation}
	\lim_{\delta \to 0^+}	 \lim_{\digamma\to\infty } \lim_{\varepsilon \to 0^+} \int_{\mathsf{2}(A_{\ell,m,n})}  \widehat{Z}[\epsilon,\delta,\varepsilon,\digamma]^* \omega(\bmalpha,\bmbeta,\bmgamma)    = \Big[\sum_{\calF\subseteq \calF(A_{\ell,m,n}) }(-1)^{|\calF|} e^{i\theta_\calF} \Big]  I_{\ell,m,n}(\bmalpha,\bmbeta,\bmgamma)
	\label{eq:misc_068}
\end{equation}
for some phases $\theta_{\calF}\in \bbR$. Indeed, the contribution of the portion of the integral over 
\begin{equation}
	\bigcup_{j=1}^N \operatorname{supp} (A_j^\circ[\epsilon,\delta,\varepsilon]^* - a_j\circ \mathrm{bd}\circ \Pi  )
\end{equation}
is $o(1)$ in the stated limit, and the integral of the rest consists of a sum of integrals each of which differs from $I_{\ell,m,n}(\bmalpha,\bmbeta,\bmgamma)$ by a branch and an $o(1)$ error. 
The branch is captured by the phase $\smash{e^{i\theta_\calF}}$ in \cref{eq:misc_068}, and the signs $\smash{(-1)^{|\calF|}}$ take care of the fact that the map $\Pi$ is orientation reversing on half of the components of the subset 
\begin{equation} 
\Pi^{-1}(A_{\ell,m,n}^\circ)\subseteq \mathsf{2}(A_{\ell,m,n}).
\end{equation}

Since the left-hand side of \cref{eq:misc_068} is independent of $\delta,\digamma,\varepsilon$, the stated equality actually holds without taking the limit:
\begin{equation}
	 \int_{\mathsf{2}(A_{\ell,m,n})}  \widehat{Z}[\epsilon,\delta,\varepsilon,\digamma]^* \omega(\bmalpha,\bmbeta,\bmgamma)    = \Big[\sum_{\calF\subseteq \calF(A_{\ell,m,n}) } (-1)^{|\calF|} e^{i\theta_\calF} \Big]  I_{\ell,m,n}(\bmalpha,\bmbeta,\bmgamma). 
	 \label{eq:misc_070}
\end{equation}
All that needs to be done is compute the $\theta_{\calF}$.

To simplify this computation, the branch of the lift $\widehat{Z}[\epsilon,\delta,\varepsilon,\digamma]$ can be chosen such that $\theta_\varnothing = 0$. 

Considering the power set $2^{\calF(A_{\ell,m,n})}$ of $\calF(A_{\ell,m,n})$ as the abelian group $\bbZ_2^{\calF(A_{\ell,m,n})}$, the map 
\begin{equation} 
	\calF(A_{\ell,m,n})\supseteq  \calF  \mapsto \theta_\calF \in \bbR, 
\end{equation} 
being a monodromy map, is an additive homomorphism. Thus, $\theta_\calF = \sum_{\mathrm{F}\in \calF }  \theta_{\{\mathrm{F}\}}$. 

The $\theta_{\{\mathrm{F}\}}$ can be computed by computing the loop $Z[\epsilon,\delta,\varepsilon,\digamma]_* [\gamma_{\varnothing,\mathrm{F}}] \in \pi_1(\calM_N(0,\infty))$. 
\begin{itemize}
	\item If $\mathrm{F} =\mathrm{F}_{\calS;0}$ for $\calS\subseteq \calI_1\cup \calI_2$, then $Z[\epsilon,\delta,\varepsilon,\digamma]_* [\gamma_{\varnothing,\mathrm{F}}]$ is homotopic to the map $\gamma: [0,1] \to \calM_N(0,\infty)$ such that the components $\gamma_j$ for $j\in \{1,\cdots,N\}\backslash \calS$ are constant and satisfy $|\gamma_j|>1/2$, and, for $j\in \calS$, 
	\begin{equation} 
		\gamma_j(t) =  
		\begin{cases}
			- 2^{-1} j^{-1} e^{2\pi i t} & (j\in \calI_1), \\  
			2^{-1} j^{-1} e^{2\pi i t} & (j\in \calI_2). 
		\end{cases}
	\end{equation}
	This is homotopic to the path where each $\gamma_j$ traverses the circle $\{|z| = 2^{-1} j^{-1}\} \subset \bbC$ choice of  (counter-clockwise, with the same starting point as before) one-by-one while the other components stay fixed.  
	From this perspective, it is apparent that $\gamma_j$ winds around $0$ once counter-clockwise, and, if $k>j$ for $j,k\in \calS$, then $\gamma_j$ winds once counter-clockwise around $\gamma_{k}$. The former contributes a monodromy of $\pi \alpha_j$, while the latter contributes a monodromy of $2\pi \gamma_{j,k}$. Thus, the overall monodromy is given by 
	\begin{equation}
		\theta_{\{\mathrm{F}\}} =  2\pi \sum_{j\in \calS} \alpha_j + 4 \pi \sum_{j,k\in \calS,j<k} \gamma_{j,k} =2 \pi \alpha_{\calS}, 
	\end{equation}
	where $\alpha_{\calS}$ was defined in \cref{eq:abdef}. 
	\item If $\mathrm{F} =\mathrm{F}_{\calS;1}$ for $\calS\subseteq \calI_2\cup \calI_3$, then the monodromy is computed analogously, with $\beta$ in place of $\alpha$, so 
	\begin{equation}
	\theta_{\{\mathrm{F}\}} = 2 \pi \sum_{j\in \calS} \beta_j + 4 \pi \sum_{j,k\in \calS,j<k} \gamma_{j,k} = 2\pi \beta_{\calS}, 
	\end{equation}
	where $\beta_{\calS}$ was defined in \cref{eq:abdef}.
	\item If $\mathrm{F}=\mathrm{F}_{\calS;\infty}$ for $\calS\subseteq \calI_1\cup \calI_3$, then $Z[\epsilon,\delta,\varepsilon,\digamma]_* [\gamma_{\varnothing,\mathrm{F}}]$ is homotopic to the map $\gamma:[0,1] \to \calM_N(0,\infty)$ such that the components $\gamma_j$ for $j\in \{1,\cdots,N\}\backslash \calS$ are constant and satisfy $\gamma_j \in (-1,+2)$ and, for $j\in \calS$, 
	\begin{equation}
		\gamma_j(t) = 
		\begin{cases}
		1-2je^{-2\pi i t} & (j\in \calI_1), \\
		2j e^{-2\pi i t} & (j\in \calI_3). 
		\end{cases}
	\end{equation}
	As before, this is homotopic to a path where only one $\gamma_j$ is moving at a time, from which it can be seen that $\gamma_j$ is winding once clockwise around each of $0,1$. 
	Additionally, for $k>j$ with $k\in \calS$, $\gamma_{k}$ is winding around $\gamma_j$ once clockwise. Finally, for $k>j$ with $j\in \calS$ and $k\notin \calS$, $\gamma_j$ is winding around $\gamma_{k}$ once clockwise as well. The total monodromy is therefore 
	\begin{equation}
	- 2\pi \sum_{j\in \calS} (\alpha_j+ \beta_j) - 4\pi \sum_{k>j, \{j,k\} \cap \calS\neq \varnothing} \gamma_{j,k} = 2\pi \zeta_{\calS}.
	\end{equation}
\end{itemize}

Thus, \cref{eq:misc_070} becomes the desired \cref{eq:fundamental}, up to an overall phase which can be eliminated with a different choice of branch.
This completes the proof of the main theorem in the $r=0,R=\infty$ case.

Because $\mathsf{2}(A_{\ell,m,n})$ is compact, the case $R\geq R_0$ for $R_0\gg 1$ sufficiently large immediately follows. 
The punctured annuli $\{z\in \bbC: r <|z| < R, z\neq 1\}$, for $r\in (0,1)$ and $R>1$, can all be deformation retracted to each other through diffeomorphisms. Consequently, the same applies to the moduli spaces $\calM_N(r,R)$, and likewise for their monodromy covers. The general case of the main theorem therefore follows from that which has already been proven.

\appendix

\section{Remark on $N=1$ case}

At first glance, it might seem surprising that the cycles $\iota_{1,0,0}(\bbS^1)$, $\iota_{0,1,0}(\bbS^1)$, and $\iota_{0,0,1}(\bbS^1)$ agree up to a choice of branch, as 
\begin{align}
	I_{1,0,0}(\alpha,\beta)&= (-1)^\alpha \int_{-\infty}^0 (-x)^\alpha (1-x)^\beta \dd x =  \frac{(-1)^\alpha \Gamma(1+\alpha)\Gamma(-1-\alpha-\beta)}{\Gamma(-\beta)}, \\ 
	I_{0,0,1}(\alpha,\beta)&= (-1)^\beta \int_1^\infty x^\alpha(x-1)^\beta \dd x = \frac{(-1)^\beta \Gamma(-1-\alpha-\beta)\Gamma(1+\beta)} {\Gamma(-\alpha)}
\end{align}
when the respective integrals are well-defined. 
By our main theorem, the right-hand sides must differ from $I_{0,1,0}(\alpha,\beta)=B(\alpha,\beta) = \Gamma(1+\alpha)\Gamma(1+\beta)/\Gamma(2+\alpha+\beta)$ by a product of trigonometric functions and a phase. That this does in fact hold is a consequence of the Gamma function's reflection identity $\sin(\pi z)\Gamma(z) \Gamma(1-z) =\pi$. The reasoning can also be run in reverse to yield a proof of the reflection identity.





\printbibliography
	
\end{document}